%% file: main.tex
\documentclass[11pt]{article}

\input{preamble}

\title{Fully Dynamic Adversarially Robust Correlation Clustering in Polylogarithmic Update Time}
\author{Vladimir Braverman\thanks{Rice University, Google Research, UT Health, and Johns Hopkins University. \texttt{email:~vova@vs.jhu.edu}}
    \and
    Prathamesh Dharangutte \thanks{Rutgers University. \texttt{email:~prathamesh.d@rutgers.edu}.}
    \and
    Shreyas Pai\thanks{Indian Institute of Technology Madras. \texttt{email:~shreyas@cse.iitm.ac.in}.}
    \and
    Vihan Shah\thanks{University of Waterloo. \texttt{email:~vihanshah98@gmail.com}.}\vspace{2pt}
    \and
    Chen Wang\thanks{Rice University and Texas A\&M University. \texttt{email:~chen.wang.research@gmail.edu}.}
}
\date{}

\begin{document}

\maketitle

\begin{abstract}
We study the dynamic correlation clustering problem with \emph{adaptive} edge label flips. In correlation clustering, we are given a $n$-vertex complete graph whose edges are labeled either $\pl$ or $\nl$, and the goal is to minimize the total number of $\pl$ edges between clusters and the number of $\nl$ edges within clusters.
We consider the dynamic setting with adversarial robustness, in which the \emph{adaptive} adversary could flip the label of an edge based on the current output of the algorithm. 
Our main result is a randomized algorithm that always maintains an $O(1)$-approximation to the optimal correlation clustering with $O(\log^{2}{n})$ amortized update time. 
Prior to our work, no algorithm with $O(1)$-approximation and $\polylog{(n)}$ update time for the adversarially robust setting was known.
We further validate our theoretical results with experiments on synthetic and real-world datasets with competitive empirical performances.
Our main technical ingredient is an algorithm that maintains \emph{sparse-dense decomposition} with $\polylog{(n)}$ update time, which could be of independent interest.
\end{abstract}

\input{intro}

\input{preliminary}

\input{edge-dynamic-alg}

\input{experiments}




\allowdisplaybreaks

\FloatBarrier

\def\shortbib{0}
\bibliographystyle{alpha}
\bibliography{references}

\clearpage
\appendix

\end{document}

%% file: preamble.tex
\pdfoutput=1
\usepackage[T1]{fontenc}
\usepackage{lmodern}
\usepackage[protrusion=true,expansion=true]{microtype}
\usepackage{amsmath,amssymb,amsfonts,amsthm}
\usepackage{subcaption}
\usepackage{graphicx}
\usepackage{fullpage}
\usepackage{lipsum}
\usepackage[english]{babel}
\usepackage[backref=page,linktocpage=true]{hyperref}
\usepackage{color}
\usepackage{wrapfig}
\usepackage{tikz}
\usetikzlibrary{decorations.pathreplacing}
\usepackage{algorithm}
\usepackage[noend]{algpseudocode}
\usepackage[framemethod=tikz]{mdframed}
\usepackage{xspace}
\usepackage{pgfplots}
\usepackage{framed}
\usepackage{thmtools}
\usepackage{thm-restate}
\usepackage{tabu}
\usepackage{fancyhdr}
\pgfplotsset{compat=1.5}
\usepackage{enumitem}

\usepackage{subcaption}
\usepackage{placeins}
\usepackage{afterpage}
\usepackage{booktabs}

\newtheorem{lemma}{Lemma}[section]
\newtheorem{proposition}{Proposition}[section]
\newtheorem{definition}{Definition}

\newtheorem{claim}[lemma]{Claim}

\newenvironment{proofof}[1]{\begin{trivlist} \item {\bf Proof
#1:~~}}
  {\qed\end{trivlist}}

\newcommand{\namedref}[2]{\hyperref[#2]{#1~\ref*{#2}}}

\newcommand{\claimlab}[1]{\label{claim:#1}}
\newcommand{\claimref}[1]{\namedref{Claim}{claim:#1}}

\newcommand{\degT}{\ensuremath{\widetilde{\deg}}\xspace}

\newcommand{\eps}{\varepsilon}

\def \calA    {\mdef{\mathcal{A}}}

\newcommand{\myqed}[1]{\let\qed\relax \hspace*{\fill} #1 \ensuremath{\square}}

\newcommand{\SDD}{\ensuremath{\texttt{SDD}}\xspace}

\newenvironment{tbox}{\begin{tcolorbox}[
		enlarge top by=5pt,
		enlarge bottom by=5pt,
		 breakable,
		 boxsep=0pt,
                  left=4pt,
                  right=4pt,
                  top=10pt,
                  arc=0pt,
                  boxrule=1pt,toprule=1pt,
                  colback=white
                  ]
	}
{\end{tcolorbox}}

\newcommand{\mdef}[1]{{\ensuremath{#1}}\xspace}  

\DeclareMathOperator*{\polylog}{polylog}

\newcommand{\set}[1]{\mdef{\left\{#1\right\}}}                        

\newcommand{\ignore}[1]{}

\newif\ifnotes\notestrue 
\ifnotes
\newcommand{\samson}[1]{\textcolor{blue}{{\bf (Samson:} {#1}{\bf ) }} \marginpar{\tiny\bf
             \begin{minipage}[t]{0.5in}
               \raggedright S:
            \end{minipage}}}
\newcommand{\david}[1]{\textcolor{purple}{{\bf (David:} {#1}{\bf ) }} \marginpar{\tiny\bf
             \begin{minipage}[t]{0.5in}
               \raggedright D:
            \end{minipage}}} 
\else
\newcommand{\samson}[1]{}
\newcommand{\david}[1]{}
\fi

\makeatletter
\renewcommand*{\@fnsymbol}[1]{\textcolor{mahogany}{\ensuremath{\ifcase#1\or *\or \dagger\or \ddagger\or
 \mathsection\or \triangledown\or \mathparagraph\or \|\or **\or \dagger\dagger
   \or \ddagger\ddagger \else\@ctrerr\fi}}}
\makeatother

\providecommand{\email}[1]{\href{mailto:#1}{\nolinkurl{#1}\xspace}}

\definecolor{mahogany}{rgb}{0.75, 0.25, 0.0}
\definecolor{darkblue}{rgb}{0.0, 0.0, 0.55}
\definecolor{darkpastelgreen}{rgb}{0.01, 0.75, 0.24}
\definecolor{darkgreen}{rgb}{0.0, 0.2, 0.13}
\definecolor{darkgoldenrod}{rgb}{0.72, 0.53, 0.04}
\definecolor{darkred}{rgb}{0.55, 0.0, 0.0}
\definecolor{forestgreenweb}{rgb}{0.13, 0.55, 0.13}
\definecolor{greencss}{rgb}{0.0, 0.5, 0.0}
\definecolor{bleudefrance}{rgb}{0.19, 0.55, 0.91}

\hypersetup{
     colorlinks   = true,
     citecolor    = darkblue,
	 linkcolor	  = darkred
}

\fancypagestyle{pg}
{
\lhead{}
\rhead{}
\cfoot{--\ \thepage\ --}

}

\newcommand{\paren}[1]{\ensuremath{\left(#1\right)}\xspace}
\newcommand{\card}[1]{\left\vert{#1}\right\vert}

\newcommand{\Vsparse}{V_{\text{sparse}}}

\newcommand{\Vup}{\ensuremath{V_{\textnormal{update}}}\xspace}

\newcommand{\sym}{\,\triangle\,}

\newcommand{\bracket}[1]{\left[#1\right]}
\newcommand{\expect}[1]{\mathbb{E}\bracket{#1}}

\newcommand{\DeltaT}{\ensuremath{\widetilde{\Delta}}\xspace}

\newcommand{\pl}{\ensuremath{(+)}\xspace}
\newcommand{\nl}{\ensuremath{(-)}\xspace}

\newcommand{\Utilde}{\ensuremath{\widetilde{U}}\xspace}

\usepackage{tcolorbox}
\tcbuselibrary{skins,breakable}
\tcbset{enhanced jigsaw}

\usepackage[normalem]{ulem}
\usepackage[compact]{titlesec}

\definecolor{DarkRed}{rgb}{0.5,0.1,0.1}
\definecolor{RURed}{rgb}{0.95,0.1,0.1}
\definecolor{DarkBlue}{rgb}{0.1,0.1,0.5}

\newtheorem{mdresult}{Result}

\definecolor{Red}{rgb}{0.9,0,0}

\usepackage{nameref}
\usepackage[nameinlink]{cleveref}

\crefname{property}{property}{Property}
\creflabelformat{property}{(#1)#2#3}

\crefname{equation}{eq}{Eq}
\creflabelformat{equation}{(#1)#2#3}

\crefname{thm}{theorem}{Theorem}
\creflabelformat{Theorem}{(#1)#2#3}

\crefname{algocf}{Algorithm}{Algorithm}
\creflabelformat{algocf}{#1#2#3}

\theoremstyle{definition}
\newtheorem{mdalg}{Algorithm}
\newenvironment{Algorithm}{\begin{tbox}\begin{mdalg}}{\end{mdalg}\end{tbox}}

%% file: intro.tex
\section{Introduction}
\label{sec:intro}
Clustering, or more simply, putting similar elements in the same group, is a fundamental task in many machine learning and data science applications. Correlation clustering \cite{bansal2004correlation} is a classic problem where we are given as input an $n$-vertex \emph{labeled complete graph} $G=(V, E^+\cup E^-)$, where each edge is labeled either positive $(+)$ or negative $(-)$, indicating pairwise similarity or dissimilarity of the two endpoints. The goal is to cluster the vertices so that the total number of disagreements is minimized. A disagreement is either a positive edge whose both endpoints are in different clusters and a negative edge whose both endpoints are in the same cluster. Unlike the $k$-clustering tasks, there is no restriction on the number of clusters in a solution as long as it minimizes the total disagreement. Correlation clustering has been extensively studied in theoretical computer science and machine learning literature (see, e.g.~\cite{bansal2004correlation,GiotisG06,BonchiGK13,BhaskaraDV18,AhmadianE0M20,Cohen-AddadLMNP21,Assadi022,Cohen-AddadFLMN22,Cohen-AddadLMP24,CaoCL0NV24}, and references therein); furthermore, correlation clustering has been widely applied to a broad spectrum of applications, including document clustering \cite{bansal2004correlation,zhang2008correlation}, social networks \cite{levorato2017evaluating,AADGW22}, community detection \cite{ShiDELM21}, image segmentation \cite{kim2011higher,kim2014image}, to name a few.

The original form of correlation clustering is defined on any \emph{fixed} labeled complete graph $G=(V, E^+\cup E^-)$. 
However, in many scenarios, the input labels may not be static and are subject to changes over time. These changes could be due to unreliable input information or the evolution of the environment. 
For instance, in social networks, we often use $\pl$ and $\nl$ to denote the ``friendly'' and ``hostile'' relationships between individuals, which would naturally evolve as time changes. 
In such situations, we would like to \emph{maintain} a solution that performs well with the most recent labeling. To this end, the trivial solution is to compute a new correlation clustering on the \emph{entire graph} upon every label update. However, solving the entire problem from scratch with only local updates appears to be very inefficient. As such, a very motivating question is whether we could maintain a solution with \emph{efficient} update times.

The above question motivates the study of \emph{dynamic} algorithms for correlation clustering \cite{BehnezhadDHSS19,ChechikZ19,BehnezhadCMT23,DalirrooyfardMM24,Cohen-AddadLMP24,BCCGM24}. In this setting, we are given a labeled complete graph whose entire set of edges is labeled $\nl$ in the beginning. For every update step, the adversary could flip an edge label from $\nl$ to $\pl$ or from $\pl$ to $\nl$. The efficiency is measured by amortized update time, defined as the total update time divided by the number of edge updates. The ``sweet spot'' of the update time is $\polylog(n)$, which is asymptotically proportional to the time needed to read a single edge update. The work of \cite{BehnezhadDHSS19,ChechikZ19} was the first to design algorithms with $3$-approximation and $\polylog(n)$ update time. Subsequently, \cite{BehnezhadCMT23,DalirrooyfardMM24} improved the update time to $O(1)$ for the same approximation guarantee, and a very recent work of \cite{BCCGM24} improved the approximation guarantee to $(3-\Omega(1))$ with $\polylog(n)$ update time. 

Almost the entire literature of dynamic clustering focuses on \emph{oblivious} adversary. In this setting, the edge label flips are prespecified by an adversary in the beginning, and the later sequence of edge updates will \emph{not} change during the algorithm updates. However, in many applications, the adversary might be \emph{adaptive}, e.g., the label flips of the next step could be \emph{dependent} on the current output of the clustering. Algorithms that work against adaptive adversaries are called \emph{adversarially robust} algorithms. The study of adversarial robustness has real-world motivations, e.g., in internet routing (see also~\cite{MironovNS11}), the packets sent by the users are based on their current experience of the current latency time. Adversarially robust algorithms have been extensively studied in various online settings (see, e.g.~\cite{MironovNS11,GilbertHSWW12,HardtW13,Ben-EliezerJWY22,BeimelKMNSS22,BernsteinBGNSS022}).

Most existing algorithms for dynamic correlation clustering are based on the classic (sequential) pivot algorithm \cite{ailon2008aggregating}. It is unclear how the pivot-based algorithms could be made to work in the adversarially robust setting. One reason is that the approximation guarantee of the pivot algorithm crucially relies on having a \emph{uniformly random permutation} of the vertices. However, in the adversarially robust setting, the adversary could learn the permutation from the clustering output and make label updates such that the sampled permutation gives a bad solution. As such, there has not been any known adversarially robust correlation clustering algorithm that maintains $O(1)$-approximation in $o(m)$ time. The only known algorithm that deals with the adversarially robust setting is a recent work by \cite{Cohen-AddadLMP24}; however, their algorithm deals with \emph{vertex updates}, and their flips from $\pl$ to $\nl$ is assumed to be random. As such, we could ask the following motivating and open question.
\begin{quote}
\centering
\textit{Can we design any adversarially robust dynamic algorithm for correlation clustering with $O(1)$-approximation and $\polylog(n)$ amortized update time?}
\end{quote}


\subsection{Our contribution}
In this paper, we answer the above question in the affirmative: we design an adversarially robust algorithm that maintains an $O(1)$-approximation to correlation clustering in $O(\log^{2}{n})$ update time. Our model is the adjacency list of the $G^{+}$ subgraph, which is consistent with the models used in \cite{Assadi022,Cohen-AddadLMP24}. To continue, we first formally introduce the model.

\begin{definition}[Adjacency list model for dynamic edge label flips]
\label{def:adj-list-pos-graph}
Let $G=(V,E^+\cup E^-)$ be a labeled complete graph with dynamic and adaptive edge label updates. The graph starts with all edges being labeled as $\nl$.
We assume that we could access the adjacency list of the positive subgraph $G^+ = (V, E^+)$ of $G$. In particular, this means in $O(1)$ time we could: $i).$ given a vertex $v\in V$, query a $\pl$ neighbor of $v$; $ii).$ given a vertex $v\in V$, query the positive degree $\deg^{+}(v)$ of $v$; and $iii).$ query an arbitrary edge $(u,v)\in E^+$. 
\end{definition}

In a labeled complete graph, the adjacency list of $G^+$ could easily be constructed by looking at the \emph{insertion} ($\nl \rightarrow \pl$) and \emph{deletion} ($\pl \rightarrow \nl$) of $\pl$ edges. The main contribution of our work is the following theorem.

\begin{mdframed}[backgroundcolor=lightgray!40,topline=false,rightline=false,leftline=false,bottomline=false,innertopmargin=2pt]
\begin{restatable}{thm}{ourmaintheorem}
\label{thm:main-alg}
    There is an adversarially robust algorithm that given a fully dynamic labeled complete graph $G=(V, E^+ \cup E^-)$ such that the edge labels are adaptively flipped, maintains an implicit representation of an $O(1)$-approximation to the optimal correlation clustering disagreement cost in $\polylog(n)$ amortized update time.
\end{restatable}
\end{mdframed}

To the best of our knowledge, our algorithm in \Cref{thm:main-alg} is the first to maintain an $O(1)$-approximation in $\polylog{(n)}$ update time against an adaptive adversary. This answers the open problem for adversarially robust correlation clustering. We now provide a comparison between our results and two closely related works.

\paragraph{Comparison with \cite{Assadi022}.} Our algorithm uses the sparse-dense decomposition of \cite{Assadi022} as a starting point. Furthermore, \cite{Assadi022} and our work share the same adjacency list model for the $G^+$ subgraph. However, our algorithm contains many non-trivial ideas beyond the scope of \cite{Assadi022}. Since \cite{Assadi022} studied the problem in the sublinear setting, where the goal is to output a correlation clustering with a \emph{static} input graph, the running time of their algorithm is $\Theta(n \log^{2}{n})$. Naively running the algorithm after every step would result in exponentially worse update time than our algorithm.

\paragraph{Comparison with \cite{Cohen-AddadLMP24}.} The work of \cite{Cohen-AddadLMP24} also studied correlation clustering with dynamic updates. However, there are two major differences: $i).$ \cite{Cohen-AddadLMP24} studied \emph{vertex} updates, where the vertices are arrived and deleted together with the adjacency list of the $\pl$ edges; $ii).$ Although the vertex insertions in \cite{Cohen-AddadLMP24} are adversarial, the deletions are random, and the randomness is essential for their correctness proof. In terms of the techniques, \cite{Cohen-AddadLMP24} is also based on the idea of sparse-dense decomposition, albeit they use a different version of sparse-dense decomposition (the version in \cite{Cohen-AddadLMNP21}). Therefore, our results are closely related to \cite{Cohen-AddadLMNP21} but not directly comparable.

\paragraph{Insertion, deletion, edges, and non-edges.} Since our model accesses the adjacency list of $G^+$ subgraph, every edge flip $\nl \rightarrow \pl$ could be viewed as an \emph{insertion} of the $\pl$ edge. Conversely, every edge flip  $\pl \rightarrow \nl$ is equivalent to a \emph{deletion} of the $\pl$ edge. As such, in the rest of the paper, we talk about \textbf{edge insertions and deletions} of the $G^+$ subgraph. We further slightly overload the term of \textbf{edge} to refer to the $\pl$ edges, and \textbf{non-edge} as the $\nl$ edges. When the context is clear, we slightly overload $G$ to denote $G^+$.

\paragraph{Experiments.} 
We empirically evaluate our algorithm against the baseline approach of running the classical pivot algorithm on the entire graph after every update. This is robust against an adaptive adversary as it ignores the previous output and generates a new result completely from scratch. However, this approach takes $\Omega(n^2)$ amortized update time in the worst case. In \Cref{sec:experiments}, we show that our algorithm performs better and has a stable clustering cost across updates for various settings across synthetic and real-world datasets.


\input{tech-overview}

\subsection{Additional Works on Correlation Clustering}

In the sequential setting, the best approximation factor that can be achieved for correlation clustering is $1.43 + \varepsilon$ \cite{CaoCL0NV24} which improves on $(1.73+\varepsilon)$-approximation of \cite{Cohen-AddadLLNFOCS2023} and $(1.994 + \varepsilon)$-approximation of \cite{Cohen-AddadLNFOCS2022}. These approaches are based on solving a linear programming relaxation and rounding the solution. This is not very amenable to the dynamic setting nor to other settings like parallel, streaming, and sublinear algorithms. One exception is \cite{CaoHSSODA2024} which gives a $(2.4+\varepsilon)$-approximation algorithm that has an efficient parallel implementation.

More combinatorial algorithms include the pivot algorithm \cite{ailon2008aggregating} which gives a $3$-approximation in expectation, and algorithms based on sparse dense decompositions, which have been recently shown to give a $\approx 1.847$-approximation \cite{Cohen-AddadLPTYZSTOC2024}. As we saw earlier, these two approaches have been quite successfully adapted to the dynamic setting \cite{BehnezhadCMT23,Cohen-AddadLMP24}. Additionally they can be implemented in the parallel, streaming, and sublinear settings as well \cite{Cohen-AddadLMNP21, Assadi022, BehnezhadCMTFOCS22, ChakrabartyM23, CambusKLPU24, BehnezhadCMTSODA23}.

The streaming setting is of particular interest since the input is accessed as a stream of edge insertions and deletions, similar to the dynamic setting. One difference, which makes the setting easier, is that in streaming the solution needs to be computed only at the end (in the dynamic setting, one has to maintain a correct solution after each update). Another difference, which makes the setting harder, is that we want the algorithm to use near-linear memory, so we cannot have arbitrary access to the entire graph (in the dynamic setting, the entire input is in memory). The streaming algorithms for correlation clustering work in a setting where only the \emph{order} of the edge insertions and deletions is determined by an adaptive adversary. To the best of our knowledge, there is no work on adversarially robust correlation clustering in the streaming setting, where the adversary can also change the edge insertions and deletions based on the current solution computed by the algorithm.

\paragraph{Concurrent and independent work on dynamic sparse-dense decomposition.}
 Recently, we have been aware that \cite{BehnezhadRWSODA25} independently and concurrently developed a dynamic sparse-dense decomposition that was used to solve the ($\Delta+1$)-coloring problem against adaptive adversaries.
 \cite{BehnezhadRWSODA25} also contains fast graph coloring techniques that are fairly involved and far beyond simply following the decomposition.
 On the other hand, we note that their sparse-dense decomposition algorithm works with the global maximum degree $\Delta$ and cannot be directly applied to correlation clustering as ours. Furthermore, our experiments on correlation clustering are entirely disjoint from their results. 
 Both our work and \cite{BehnezhadRWSODA25} demonstrate the wide range of applications for dynamic sparse-dense decomposition, which can be a technique of independent interest.

\paragraph{Open problems.} Our work settles the constant approximation fully dynamic algorithms for correlation clustering with the \emph{edge update} model. 
We believe the following problems can serve as interesting open directions to further explore.
\begin{itemize}
    \item \textbf{Fully vertex dynamic $O(1)$-approximation correlation clustering.} \cite{Cohen-AddadLMP24} gave an $O(1)$-approximation for vertex dynamic correlation clustering with $\polylog(n)$ update time; however, the deletion operations in their model only allow the adversary to uniformly at random sample a vertex for deletion. In their paper, they constructed an example showing that adversarial deletion will lead to $\Omega(n)$ amortized update time for their algorithm. Therefore, it is very interesting to ask whether we can obtain $O(1)$-approximation for correlation clustering with $\polylog(n)$ (or even $o(n)$) update time that supports \emph{adversarial deletions}.
    \item \textbf{Deterministic algorithms for efficient correlation clustering.} Similar to this work, the majority of work for linear and sublinear time algorithms for correlation clustering relies on randomization (see, e.g.,~\cite{Assadi022,AssadiSW23,Cohen-AddadLMP24,CaoCL0LNTVY025}). Therefore, it is an interesting question to ask whether we can obtain efficient $O(1)$-approximation \emph{deterministic} algorithms dynamic (or even offline) correlation clustering\footnote{In a previous version of this paper, \Cref{prop:sdd-alg} stated that the sparse-dense decomposition in \cite{Assadi022} can be solved in $O(m)$ time, which was incorrect. This mistake does not affect any conclusion of this paper (we never used this in our algorithms).}.
\end{itemize}

%% file: tech-overview.tex
\subsection{Technical overview}
\label{sec:tech-overview}
The starting point of our techniques is the sparse-dense decomposition of graphs that was recently introduced to the correlation clustering literature \cite{Cohen-AddadLMNP21,Bun0K21,Assadi022,Cohen-AddadFLMN22,AssadiSW23,Cohen-AddadLMP24}. 
Roughly speaking, graph sparse-dense decomposition is a family of techniques that given a $n$-vertex graph $G=(V,E)$, divide the vertices into \emph{almost-cliques}, defined as vertices that could be ``bundled'' together with few internal non-edges and outgoing edges, and \emph{sparse vertices}, defined as vertices with low degrees or spread connectivity to various almost-cliques. \cite{Cohen-AddadLMNP21,Assadi022} showed that we could obtain an $O(1)$-approximation for correlation clustering by clustering vertices based on a sparse-dense decomposition on the $G^+$ subgraph. Furthermore, \cite{Assadi022} proved that we could compute the sparse-dense decomposition in $O(n \log^{2}{n})$ time.

The results of \cite{Assadi022} already implied an adversarial robust dynamic algorithm with $O(n \log^{2}{n})$ amortized update time: we could simply restart the algorithm after every edge update. However, running the algorithm from scratch appears to be very ``wasteful''. Previous work like \cite{Cohen-AddadLMP24} has shown that the decomposition is fairly robust and resistant to a small number of edge updates. As such, we could hope to take advantage of robustness and update clustering based on the existing structure of the decomposition. However, the task is highly non-trivial since the targeted running time, i.e., $\polylog{(n)}$, is almost exponentially smaller than the time to run the algorithm. 

On a very high level, the main idea for our algorithm to achieve $\polylog{n}$ amortized update time is by running sparse-dense decomposition locally in $N[u]$ for a vertex $u$ after the degree of $u$ changes by a small constant factor. For a relatively easy example, consider the scenario that all vertices are of degree $\Theta(n)$. In this case, we could perform sparse-dense decomposition on the entire graph after the updates of $0.01 \, n$ edges, and the amortized running time would become $O(\log^2{(n)})$ since we have enough updates to ``take charge'' for the updates. 


\begin{figure}
  \centering
  \begin{subfigure}[t]{.4\linewidth}
    \centering\includegraphics[width=\linewidth]{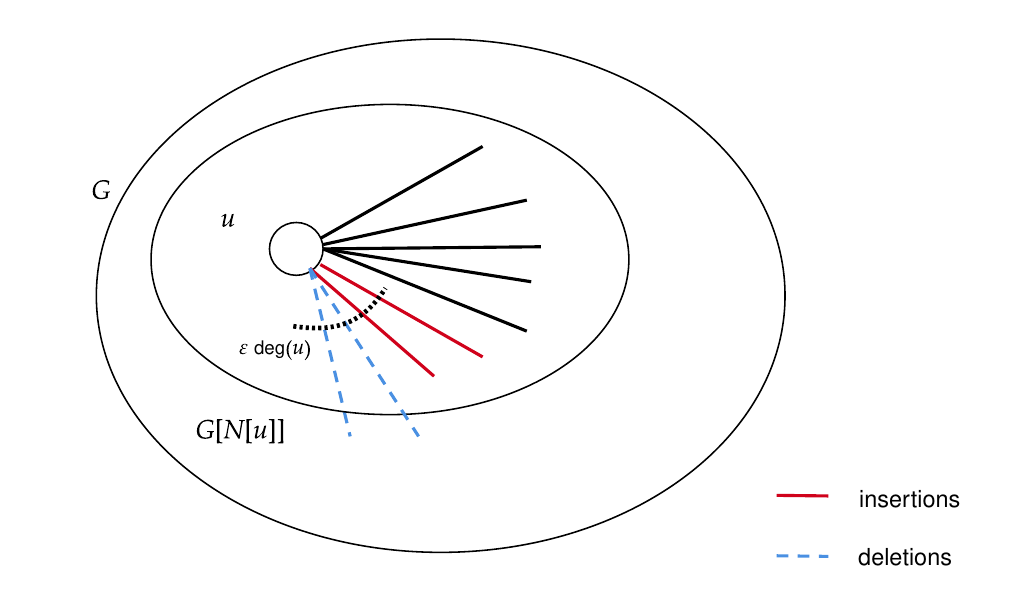}
    \caption{}
    \label{fig:overall-strategy}
  \end{subfigure}
  \begin{subfigure}[t]{.4\linewidth}
    \centering\includegraphics[width=\linewidth]{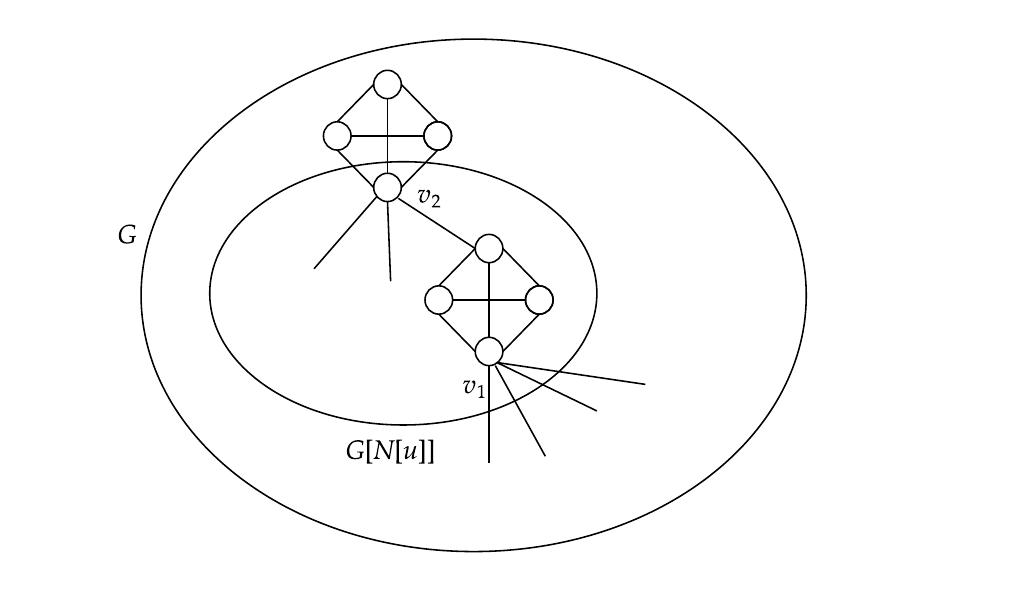}
    \caption{}
    \label{fig:local-decomp-difficulty}
  \end{subfigure}
  \caption{An illustration of the overall update strategy and the discrepancy between the local and global sparse-dense decomposition. In \Cref{fig:overall-strategy}, we run algorithm updates on $G[N[u]]$ once the number of insertions and deletions on $u$ has exceeded $\eps\cdot \deg(u)$. In \Cref{fig:local-decomp-difficulty}, note that vertex $v_1$ belongs to an almost-clique in $G[N[u]]$, but is sparse globally. $v_2$ is a sparse vertex in $G[N[u]]$, but it belongs to another almost-clique.}
\end{figure}


There are however significant technical challenges to make the above idea work for the general case. For general vertex degrees, if we run the decomposition on the induced subgraph of $N[u]$ after $0.01\, \deg(u)$ updates, we could still guarantee $O(\log^2{n})$ amortized running time\footnote{Careful readers might have noticed that it is not clear how could we run sparse-dense decomposition on \emph{induced subgraphs} with the adjacency list model. We provide a discussion in \Cref{sec:tech-subroutines}.}. However, the correctness does not follow if we simply update the global decomposition with the new local one. In particular, if a vertex $v$ is sparse locally, it might still belong to an almost-clique in the global graph. Conversely, if $v$ belongs to a local almost-clique, it might be a sparse vertex in the global graph or belong to a larger almost-clique. Furthermore, the changes in the local decomposition could affect almost-cliques beyond the scope of $N[u]$, which further complicates the issue. An illustration of the overall strategy and the primary difficulty can be found in \Cref{fig:overall-strategy} and \Cref{fig:local-decomp-difficulty} (see \Cref{fig:AC-lose-vertex} for how the updates affect almost-cliques beyond $N[u]$). 

Fortunately, we could efficiently \emph{test} whether the vertices in the local sparse-dense decomposition are aligned with their global properties. In particular, for any local almost-clique $K$, we show that a vertex $v\in K$ is ``valid'' if and only if $\deg(v)$ (in the global graph $G$) is comparable to $\card{K}$. In all other cases, i.e., $v$ is globally sparse or it belongs to a much larger almost-clique, there must be many edges going out of $K$. As such, we could update the decomposition only with the ``valid'' vertices, and simply treat vertices as in the previous stage of the decomposition. For the sparse vertices, we show that we could test whether a vertex is sparse in $O(\log^{2}{n})$ time. As such, we could test whether existing vertices in the almost-cliques (not limited to the new almost-cliques) become sparse, and ``pull out'' the vertex accordingly. 



There is however yet another concern: an almost-clique $K$ might lose many vertices due to the eventual removal of the vertices; furthermore, the removed vertices are quite sparse, which will make the remaining vertices inside $K$ sparse. This piece of information might not be captured by the algorithm updates of any \emph{single} vertex $v$ and $N[v]$. To handle the challenge, we track the number of sparse vertices that are \emph{removed} from the almost-clique. If there are $0.01$ faction of the vertices removed, we simply dismantle the entire almost-clique and make every vertex sparse.

The above procedure enjoys a fast amortized update time: for every dismantle operation of an almost-clique $K$, it takes $\card{K}$ time; however, we must have at least $\Omega(\card{K})$ edge updates to make this happen. Of course, the adversary could potentially make the vertices that are initially being ``pulled out'' of an almost-clique dense again. However, in such a case, a local sparse-dense decomposition will be triggered, and new almost-cliques will be formed. As such, we could guarantee both the update efficiency and the correctness of the decomposition.
\afterpage{
\begin{figure}
  \centering
  \begin{subfigure}[t]{.38\linewidth}
    \centering\includegraphics[width=\linewidth]{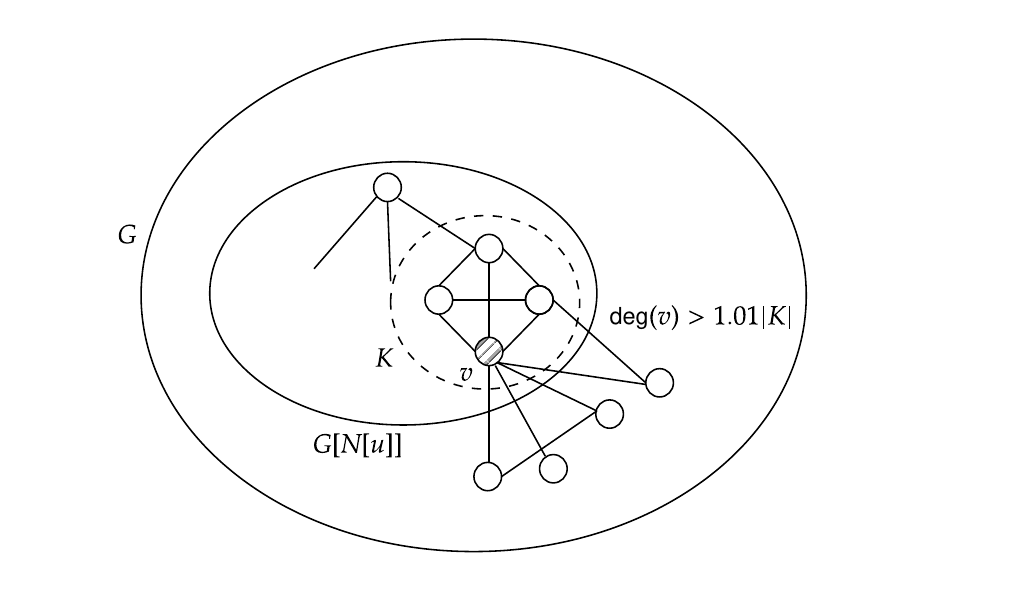}
    \caption{}
    \label{fig:local-AC-check}
  \end{subfigure}
  \begin{subfigure}[t]{.38\linewidth}
    \centering\includegraphics[width=\linewidth]{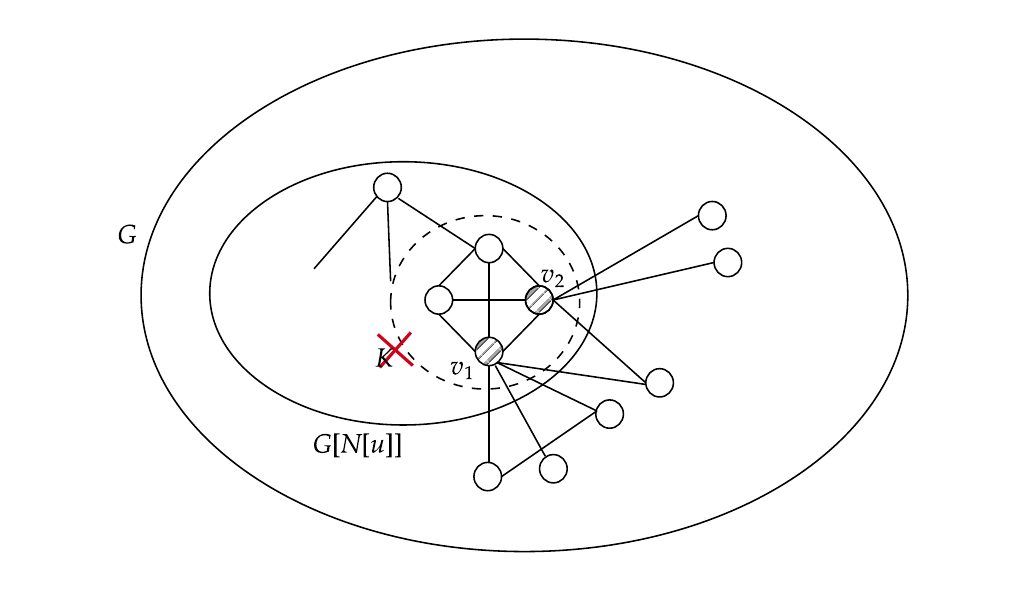}
    \caption{}
    \label{fig:local-AC-not-form}
  \end{subfigure}
  \centering
  \begin{subfigure}[t]{.38\linewidth}
    \centering\includegraphics[width=\linewidth]{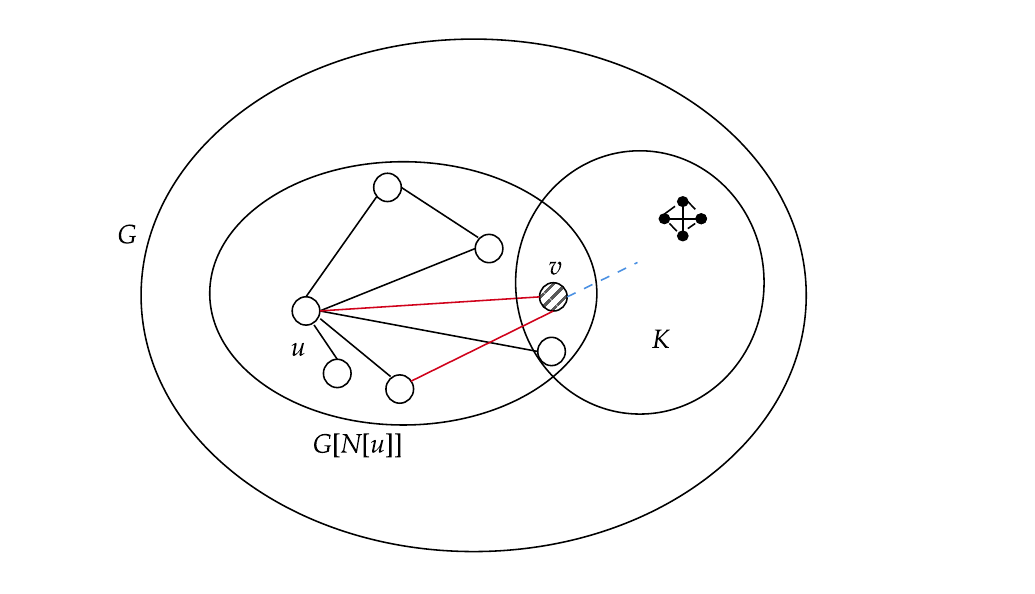}
    \caption{}
    \label{fig:AC-lose-vertex}
  \end{subfigure}
  \begin{subfigure}[t]{.38\linewidth}
    \centering\includegraphics[width=\linewidth]{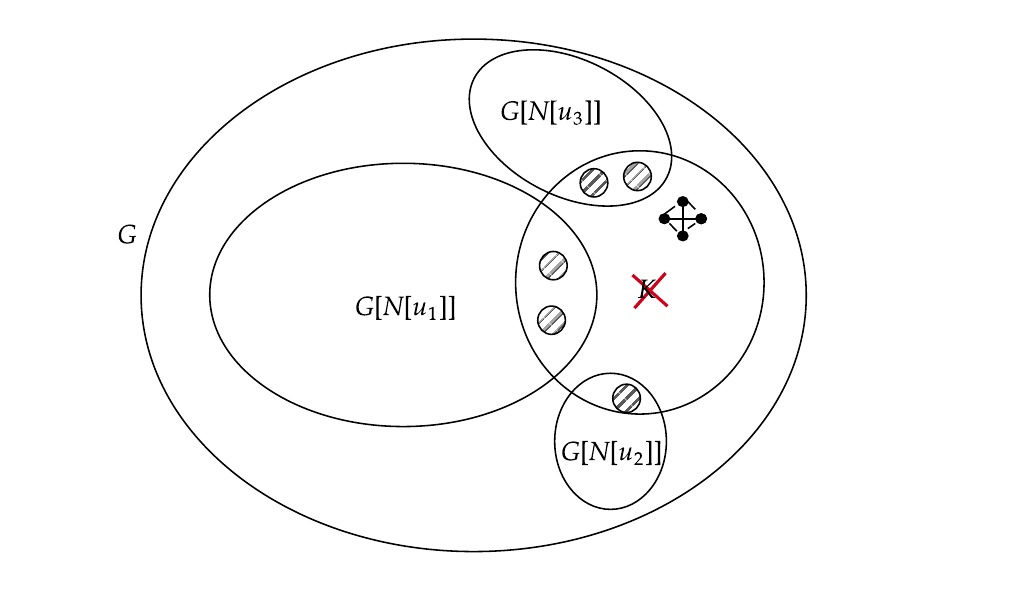}
    \caption{}
    \label{fig:AC-dismantle}
  \end{subfigure}
  \caption{An illustration of the handling of locally dense and sparse vertices. \Cref{fig:local-AC-check}: if the degree of vertex $v$ (in a local almost-clique $K$) is too high, it cannot be actually an almost-clique vertex; \Cref{fig:local-AC-not-form}: if $K$ contains too many vertices that are \emph{not} valid, we simply do not form the almost-clique; \Cref{fig:AC-lose-vertex}: since we test the sparsity for all vertices in $N[u]$, a vertex $v$ that belongs to another almost-clique $K$ could be recognized as sparse; \Cref{fig:AC-dismantle}: if too many vertices are removed from $K$ (and if $K$ is not reformed into another almost-clique), we make every vertex in $K$ sparse.}
\end{figure}
}


\FloatBarrier

%% file: preliminary.tex
\section{Preliminaries}
\label{sec:prelim}
We introduce the basic notation and the notion of sparse-dense decomposition in this section.

\paragraph{Notation.} As standard, we denote a unweighted graph $G=(V,E)$ with $V$ as the set of the vertices and $E$ as the set of the edges. For any subset of vertices, $A \subseteq V$, we use $\bar{A}=V \setminus A$ to denote the complementary set of vertices in $G$. Let $\calA = \{A_{1},A_{2},\cdots\}$ be a family of collections of vertices. We say that $\calA$ is a \emph{partition} of $G$ if $i).$ $\cup_{i} A_i = V$, and $ii).$ for any pair $V_i, V_j \in \calA$, $V_i\cap V_j =\emptyset$.

For any vertex $v \in V$, we use $N(v)$ to denote the set of \emph{neighbors} of $v$, and $N[v]$ to denote the augmented neighborhood, i.e., $N[v]:=N(v)\cup \{v\}$. For a fixed set of vertices $U$, we use $N_U(v)$ to denote the set of neighbors of $v$ that are in $U$, i.e., $N_{U}(v):=N(v)\cap U$.

\paragraph{The sparse-dense decomposition.} 
We now introduce the \emph{sparse-dense decomposition} that gives a \emph{partition} of the graph. The vertices are divided into \emph{almost-cliques}, which roughly means the vertices that can be ``bundled together'', and \emph{sparse vertices}, which roughly means the vertices that cannot join any almost-cliques. The technique has a rich history in theoretical computer science, and we use the version of \cite{Assadi022} for the purpose of correlation clustering.

\begin{definition}[\cite{Assadi022}]
\label{def:sdd}
Given a graph $G=(V,E)$, an $\eps$-sparse-dense decomposition $\{\Vsparse, K_1, \ldots, K_k\}$, denoted as $\SDD_{G,\eps}$, is a partition of $G$ consisting of:
\begin{itemize}
\item \textbf{Sparse vertices} $\Vsparse$: There exists an absolute constant $\eta_0$ such that for any vertex $v\in \Vsparse$, either $N(v)=\emptyset$, or there are at least $\eta_0 \cdot \eps \cdot \deg{v}$ neighbors $u$ such that $\card{N(v) \sym N(u)} \geq \eta_0 \cdot \eps \cdot \max\set{\deg{(u)},\deg{(v)}}$. 
\item \textbf{Dense vertices partitioned into \underline{almost-cliques} $K_1,\ldots,K_k$}: For every $i \in [k]$, each $K_i$ has the following properties. Let 
 $\Delta(K_{i})$ be the maximum degree of the vertices in $K_i$ (the degree is counted in $G$), then: 
\begin{enumerate}[leftmargin=20pt, label=$\roman*)$.]
\item\label{dec:non-neighbors}  Every vertex $v\in K_{i}$ has at most $\eps\cdot \Delta(K_{i})$ \emph{non-neighbors} inside ${K_{i}}$;
\item\label{dec:neighbors} Every vertex $v\in K_{i}$ has at most $\eps\cdot \Delta(K_{i})$ \emph{neighbors} outside ${K_{i}}$;
\item\label{dec:size} Size of each $K_i$ satisfies $(1-\eps)\cdot \Delta(K_{i}) \leq |K_{i}| \leq (1+\eps)\cdot \Delta(K_{i})$.
\end{enumerate}
\end{itemize}
\end{definition}
When we refer to a sparse-dense decomposition $\SDD_{G,\eps}$, we mean a partition of $G$ with the labels on the vertices, i.e., for each vertex $u\in \SDD_{G,\eps}$, there is a corresponding label indicating whether it is a sparse vertex, and if not, which almost-clique it belongs to.

\paragraph{Algorithms for sparse-dense decomposition.} The main message of \cite{Assadi022} is a \emph{randomized} algorithm that computes a $\SDD_{G,\eps}$ in $O(n\polylog{n})$ time. Furthermore, in their proof of the existence of the sparse-dense decomposition, it is also implied that there exists a \emph{deterministic} algorithm that computes an $\eps$-sparse-dense decomposition in polynomial time. The formal statement is given as follows.
\begin{proposition}[\cite{Assadi022}]
\label{prop:sdd-alg}
For any input graph $G=(V,E)$ and every $\eps<\frac{1}{64}$, there always exists a sparse-dense decomposition of $G$. Moreover, there exist
\begin{enumerate}[label=\roman*).]
\item a deterministic algorithm that, given an unweighted graph $G=(V,E)$, computes an $\eps$-sparse-dense decomposition in polynomial time.
\item a randomized algorithm that, given an unweighted graph $G=(V,E)$, with high probability computes an $\eps$-sparse-dense decomposition in $O(n \log^2{n})$ time. The running time is deterministic, and the randomness is only over the success probability.
\end{enumerate}
\end{proposition}


We could actually `relax' the notion of sparse-dense decomposition by allowing the ``sparse'' and ``dense'' vertices to have different parameters. In particular, we define the following notion of sparse and dense vertices in the same spirit of \cite{AssadiSW23}.
\begin{definition}[cf. \cite{Assadi022,AssadiSW23}]
\label{def:sparse-dense-vertices}
Given a graph $G=(V,E)$ and a sparse-dense decomposition $\SDD_{G, \eps}$ with the corresponding constant $\eta_0$, we define the $\eps'$-sparse vertices and $\eps$-dense vertices as follows.
\begin{itemize}
\item We say that a vertex $v$ is \emph{$\eps'$-sparse} if either $N(v)=\emptyset$, or there exists at least $\eta_0 \cdot \eps' \cdot \deg{v}$ neighbors $u$ such that $\card{N(v) \sym N(u)} \geq \eta_0 \cdot \eps' \cdot \max\set{\deg{(u)},\deg{(v)}}$. 
\item We say that a vertex $v$ is \emph{$\eps$-dense} if there exists an almost-clique $K \subseteq \SDD_{G, \eps}$ such that $v \in K$.  
\end{itemize}
\end{definition}
In other words, the sparse and dense vertices in \Cref{def:sparse-dense-vertices} are allowed to operate with different values of $\eps$. The relaxation gives us more power to deal with dense and sparse vertices, respectively. Furthermore, as long as both $\eps$ and $\eps'$ are constants, the sparse-dense decomposition still gives an $O(1)$ approximation. The formal statement is given as follows.
\begin{proposition}[cf. \cite{Assadi022,AssadiSW23}]
\label{prop:sdd-approx}
Let $G=(V, E^+\cup E^-)$ be a correlation clustering instance, and let $\SDD'$ be a sparse-dense decomposition on $G^+=(V, E^+)$ such that every vertex $v\in \Vsparse$ is $\eps'$-sparse and every almost-clique $K_i$ is $\eps$-dense. Then, for any constant $\eps$ and $\eps'$, if we
\begin{itemize}
\item cluster each almost-clique $K_i$ in the same cluster;
\item cluster each sparse vertex as a singleton cluster
\end{itemize}
to form a correlation clustering, we can get an $O(1)$-multiplicative approximation to the optimal correlation clustering cost.
\end{proposition}

On a very high level, our dynamic algorithms maintain a sparse-dense decomposition that ensures the sparse vertices are at least ``somehow sparse'' and the almost-cliques are ``sufficiently dense''. The main challenge is to ensure these properties in an adversarial-robust manner.

%% file: edge-dynamic-alg.tex
\section{Our Main Algorithm}
\label{sec:dynamic-alg-edge-update}
We present our main algorithm and the analysis in this section. We recall our main theorem statement as follows.

\ourmaintheorem*
The key aspects to prove are the approximation guarantees and the amortized update time, i.e,
\begin{enumerate}
\item $O(1)$-approximation at any time.
\item $\polylog{n}$ update time per insertion and deletion. 
\end{enumerate}

The presentation of our algorithm is as follows. We will first show some technical subroutines that are very helpful for the main algorithm in \Cref{sec:tech-subroutines}. Then, we present our main algorithm in \Cref{subsec:main-alg} before giving the analysis in \Cref{subsec:analysis}. For ease of presentation, we assume $\eta_0 = 1$ for the sparse-dense decomposition in \Cref{def:sdd} -- we can always rescale the parameter. We further use $\SDD_{G}$ to denote the sparse-dense decomposition we maintain for the global graph, and $\SDD_{U}$ for the local decomposition.

\input{subroutines}

\paragraph{A sparse-dense decomposition on induced subgraphs.} We now discuss performing sparse-dense decomposition on \emph{induced subgraphs} of a given set of vertices $U$. Essentially, the algorithm follows from \Cref{prop:sdd-alg} (\cite{Assadi022}) with a slight modification. The main proposition we use is as follows.
\begin{proposition}
\label{prop:sdd-local-alg}
Let $G=(V,E)$ be an input graph, and let $u\in V$ be a fixed vertex. Furthermore, let $U=N(u)\cup \{u\}$ be the augmented neighborhood of $u$. Then, for every $\eps<\frac{1}{64}$, there exists a randomized algorithm that computes an $\eps$-sparse-dense decomposition on the induced subgraph $G[U]$ in $O(\frac{\card{U} \cdot \log^2{n}}{\eps^2})$ time such that
\begin{enumerate}
\item\label{line:sdd-local-sdd-prop} Every almost-clique satisfies the properties as in \Cref{def:sdd} of $G[U]$.
\item\label{line:sdd-local-u-prop} If vertex $u$ is \emph{not} at least $\eta_{0}\cdot \eps$-sparse (in $G$), then $u$ belongs to an almost-clique of the output.
\end{enumerate}
\end{proposition}
Note that some sparse vertices obtained by \Cref{prop:sdd-local-alg} might not satisfy the property as prescribed by \Cref{def:sdd}. \Cref{prop:sdd-local-alg} only guarantees that if $u$ is dense, it will form an almost-clique -- we will eventually show that the guarantee is sufficient for our main algorithm.

Proving \Cref{prop:sdd-local-alg} from scratch would involve an unnecessary repetition of all the steps in \cite{Assadi022}. As such, we provide a discussion on what needs to be changed for the algorithm of \Cref{def:sdd}. Essentially, for a graph $G=(V,E)$ with $n$ vertices, the sparse-dense decomposition algorithm computes the output based on the following \emph{samples}.
\begin{itemize}
\item For each vertex $v\in V$, sample $O(\log{n}/\eps^2)$ neighbors.
\item For each vertex $v\in V$, sample the vertex with probability $\min\{C\cdot \frac{\log{n}}{\deg(v)}, 1\}$ for some constant $C$; if $v$ is sampled, take all the neighbors of $v$.
\end{itemize}
Here, the samples in the first bullet for each vertex are used for \emph{identifying sparse vertices} and \emph{testing the symmetric difference of the neighborhoods for any pair of vertices $(u,v)$}.
In contrast, the samples in the second bullet are used to \emph{form almost-cliques} only, i.e., if $v$ is at least $\eta_0 \cdot \eps$-sparse, we do \emph{not} need to sample $v$ for the second bullet.
We need to show how to simulate both of the sampled sets in the induced subgraph and discuss how to implement the algorithm based on the sampled sets.

We first identify a set of vertices $\Utilde$ that has at least $1/2$ fraction of their degree in the induced subgraph $G[U]$. 
To obtain such a set of vertices, we can sample $O(\log n)$ neighbors for each vertex $v\in U$, and add $v$ to $\Utilde$ if a $2/3$ fraction lies in $G[U]$.
For all vertices in $\Utilde$, we can simulate the first bullet by picking a constant factor of more samples and then arguing that a constant fraction of those lie in $G[U]$. 
For all vertices in $U\setminus \Utilde$, we simply let them be \emph{sparse vertices}. We observe that by checking the property of \Cref{def:sdd}, if $u$ is contained in an almost-clique $K$, then $v\in U\setminus \Utilde$ cannot belong to $K$.

We now discuss how to obtain the samples of the second bullet for vertices in $\Utilde$. We would face the challenge that we do \emph{not} know the degree for every vertex $v$ in the induced subgraph $G[U]$. Nevertheless, observe that if we replace $\deg_{G[U]}(v)$ in the sampling probability with $\widetilde{\deg}_{G[U]}(v)$ which is a constant approximation of $\deg_{G[U]}(v)$, the algorithm is asymptotically the same. Since we care only about vertices with a constant fraction of their degree in $G[U]$, we can use the degree $\deg(v)$ of vertex $v$ in $G$ in the sampling probability instead of $\deg_{G[U]}(v)$ and increase the constant $C$ appropriately so that the algorithm remains the same.
As such, the algorithm for \Cref{prop:sdd-local-alg} could be described as follows.
\begin{tbox}
\begin{enumerate}
\item For each vertex $v\in U$, sample $50\log{n}$ neighbors uniformly at random into set $C(v)$.
\item For each $v\in U$, if $\card{C(v)\cap U}\geq \frac{2}{3}\cdot \card{C(v)}$, then add $v$ to $\Utilde$.
\item For every vertex $v\in U\setminus \Utilde$, let $v$ be a sparse vertex of $G[U]$.
\item For every vertex $v\in \Utilde$, use $\deg(v)$ to sample vertex $v$ with probability $\min\{C\cdot \frac{\log{n}}{\deg(v)}, 1\}$.
\item Follow all steps for vertices in $\Utilde$ as in the algorithm of \Cref{prop:sdd-alg}.
\end{enumerate}
\end{tbox}
By a Chernoff bound argument, we can show that if $\card{C(v)\cap U}\geq \frac{2}{3}\cdot \card{C(v)}$, then with high probability, we have that at least $\frac{1}{2}$ fraction of the neighbors of $v$ are in $U$. As such, we could get a constant approximation of the sampling probability by using $\deg(v)$ to sample. The guarantee of line~\ref{line:sdd-local-sdd-prop} simply follows from \Cref{prop:sdd-alg}, and the guarantee of line~\ref{line:sdd-local-u-prop} follows since $u$ always satisfy the property that $\card{C(u)\cap U}\geq \frac{2}{3}\cdot \deg(u)$, and all the sparse vertices in $N[u]$ cannot be in an almost-clique that contains $u$.


\subsection{The main algorithm}
\label{subsec:main-alg}
We now introduce our main algorithm. We first give our algorithm that \emph{produces} almost-cliques: this algorithm serves as the main ``workhorse'' to merge vertices to almost-cliques in the updates.
\begin{Algorithm}
\label{alg:clique-generation}
\textbf{An algorithm that computes almost-cliques for a given vertex $u$.}\\
\textbf{Input:} A graph $G=(V,E)$; a sparse-dense decomposition $\SDD_{G,\eps}$; a vertex $u\in V$ in $G$.
\smallskip
\begin{enumerate}
\item Let $U$ be the collection of vertices in $N[u]$.
\item\label{line:SDD-alg-invoke} Run the sparse-dense decomposition algorithm on the \emph{induced subgraph} $G[U]$ to get $\SDD_{L}$ with the algorithm of \Cref{prop:sdd-local-alg} and parameter $\eps/2$. 
\item\label{line:AC-forming} For each vertex $v \in U$ such that $v \in K$ for some almost-clique $K$, if $\deg(v)\leq (1+2\eps)\cdot \card{K}$, then mark $v$ as a \emph{valid} member of $K$. 
\item\label{line:AC-checking} For each almost-clique $K \in \SDD_{L}$, if $K$ has $(1-\eps)$-fraction of vertices that are valid member of $K$, let the valid members be the almost-clique $K$.
\item\label{line:AC-add-vertex} For each almost clique $K\in \SDD_{L}$, run \Cref{alg:AC-dense-test}, and obtain the new almost-cliques in $\SDD_{G, \eps}$.
\item For all other vertices, use the sparse and dense decomposition in $\SDD_{G, \eps}$. Update the $\SDD_{G, \eps}$ with the new almost cliques and record the number of vertices.
\end{enumerate}
\end{Algorithm}

We are now ready to present the main procedures of our algorithm as \Cref{alg:dynamic-alg}. Note that we use a different numbering scheme for the steps so that we can easily distinguish the line numbers in \Cref{alg:dynamic-alg} vs. the intermediate algorithms.
\begin{Algorithm}
\label{alg:dynamic-alg}
\textbf{An algorithm that maintains an $O(1)$-approximation for correlation clustering anytime.}\\
\textbf{Input:} A graph $G=(V,E)$ with one edge insertion update per time. 
\smallskip

\begin{enumerate}[label=(\Roman*)]
\item For each vertex $v \in V$, we maintain
\begin{enumerate}[label=(\roman*)]
\item a counter $c_v$ to record the number of edge insertions;
\item a sparse-dense decomposition $\SDD_{G,\eps}$ initialized with each vertex as a singleton sparse vertex;
\item a vertex degree of a past time $\degT(v)$ initialized with $\degT(v)\leftarrow 0$.
\end{enumerate}
\item For each almost-clique, we maintain the number of vertices in the component (see implementation details below).
\item Upon the insertion or deletion of an edge $(u,v)$:
\begin{enumerate}[label=(\roman*)]
\item\label{line:u-updates} If $c_u \geq (\eps/10) \cdot \degT(u)$, run the updates of \Cref{alg:clique-generation} with vertex $u$ and $\SDD_{G,\eps}$, and reset $c_w \gets 0$ and let $\degT(w) \gets \deg(w)$ for all $w$ in almost-cliques. 
\item\label{line:v-updates}  If $c_v\geq (\eps/10) \cdot \degT(v)$, run the updates of \Cref{alg:clique-generation} with vertex $v$ and $\SDD_{G,\eps}$, and reset $c_w \gets 0$ and let $\degT(w) \gets \deg(w)$ for all $w$ in almost-cliques.
\item Let $U=N[u]$ if the updates in line~\ref{line:u-updates} are executed and $U=\emptyset$ otherwise. Define $\widetilde{U}$ for $v$ in the same manner using line~\ref{line:v-updates}.
\item\label{line:sparse-vertex-removal} For each vertex in $v \in U \cup \widetilde{U}$, run \Cref{alg:vertex-dense-test} in parallel and determine whether the vertices should be put into $\Vsparse$. Update the graph partition $\SDD_{G, \eps}$, reset $c_v \gets 0$, and let $\degT(v) \gets \deg(v)$.
\item\label{line:almost-clique-dismantle} For all almost cliques $K$ that contain vertices in $U \cup \widetilde{U}$:
\begin{enumerate}[label=(\alph*)]
\item If more than $\eps \card{K}$ vertices are removed from $K$ (either by vertices becoming sparse or joining other almost-cliques), then put every remaining vertex to $\Vsparse$. 
\item For each vertex being put into $\Vsparse$, reset $c_v \gets 0$ and let $\degT(v) \gets \deg(v)$ for $v \in K$.
\item Update the graph partition $\SDD_{G, \eps}$. 
\end{enumerate}
\end{enumerate}
\item Maintain the clustering as the sparse-dense decomposition at any time, i.e., keep each almost-clique as a separate cluster and each sparse vertex as a singleton cluster.
\end{enumerate}
\end{Algorithm}
In other words, for each vertex $u\in V$ with a recorded degree $\degT(u)$, \Cref{alg:dynamic-alg} performs an update step for every $\eps\cdot \degT(u)$ updates (insertions or deletions) on $u$. For each update, the algorithm first runs a sparse-dense decomposition to recognize \emph{local almost-cliques} (note that the sparse vertices are not updated in lines~\ref{line:u-updates} and \ref{line:v-updates}). Subsequently, the algorithm checks whether the newly-formed almost-cliques should indeed be almost-cliques and adds new sparse vertices if possible. Finally, during this process, some global almost-cliques may lose some vertices; and if the number becomes significant, we simply dismantle the almost-clique and mark all remaining vertices as sparse vertices.

We analyze \Cref{alg:dynamic-alg} for the rest of this section.


\subsection{The Analysis}
\label{subsec:analysis}
We will first show that the amortized update time of \Cref{alg:dynamic-alg} is at most $O(\frac{\log^{2}{n}}{\eps^3})$ and since we choose $\eps = \Theta(1)$ we get the desired update time of $\polylog(n)$. Then, we will show that \Cref{alg:dynamic-alg} maintains a decomposition such that all sparse vertices are at least $\eps/4$-sparse and all almost-cliques are at most $120\eps$-dense. We note that we picked the large constants to simplify the calculations but the actual guarantees should be much better than the constants we use in the analysis.

\subsubsection*{Update Time Analysis} 
Our key lemma for the update time analysis is as follows.

\begin{lemma}
\label{lem:dynamic-alg-update-time}
The amortized update time for each edge insertion/deletion update in \Cref{alg:dynamic-alg} is $O(\frac{\log^{2}{n}}{\eps^3})$.
\end{lemma}
The proof of \Cref{lem:dynamic-alg-update-time} is split into two parts: we bound the amortized running time for the updates of lines~\ref{line:u-updates}, \ref{line:v-updates}, and \ref{line:sparse-vertex-removal} and line~\ref{line:almost-clique-dismantle} respectively as follows. 

\begin{lemma}
\label{lem:update-time-vertex}
The amortized update time for each edge insertion/deletion update in lines~\ref{line:u-updates}, \ref{line:v-updates}, and \ref{line:sparse-vertex-removal} of \Cref{alg:dynamic-alg} is $O(\frac{\log^{2}{n}}{\eps^3})$.
\end{lemma}
\begin{proof}
For each vertex $v\in V$, the change of clustering only happens after $\eps \cdot \degT(v)$ updates. Therefore, it suffices to show that the updates in lines~\ref{line:u-updates}, \ref{line:v-updates}, and \ref{line:sparse-vertex-removal} can be computed in time $O(\degT(v) \cdot \frac{\log^{2}{n}}{\eps^2})$ once updates are invoked. Concretely, each update affects exactly two vertices, and the amortized update time for each edge update will be
\begin{align*}
\frac{O(\degT(v) \cdot {\log^{2}{n}}/{\eps^2})}{\eps \cdot \degT(v)} = O\left(\frac{\log^{2}{n}}{\eps^3}\right),
\end{align*}
which is desired by the lemma statement.

Note that any clustering changes can only be triggered by line~\ref{line:u-updates} and line~\ref{line:v-updates} (lines~\ref{line:sparse-vertex-removal} will \emph{not} be invoked if $U \cup \widetilde{U} = \emptyset$). For the computations inside line~\ref{line:u-updates} and line~\ref{line:v-updates}, we claim that the running time is $O(\degT(v) \cdot \frac{\log^{2}{n}}{\eps^2})$. To see this, assume w.log. that we only deal with $U$. The induced subgraph $G[U]$ contains only $\deg(v)\leq (1+\frac{\eps}{100})\cdot \degT(v)$ vertices. Therefore, by \Cref{prop:sdd-local-alg}, running the local sparse-dense decomposition on $G[U]$ (line~\ref{line:SDD-alg-invoke} of \Cref{alg:clique-generation}) takes time at most 
\[O\left(\card{U}\cdot \frac{\log^{2}{n}}{\eps^2}\right) = O\left(\deg(v) \cdot \frac{\log^{2}{n}}{\eps^2}\right) = O(\degT(v) \cdot {\log^{2}{n}}/{\eps^2}).\] 
For the slightly more involved case of line~\ref{line:AC-add-vertex} that runs \Cref{alg:AC-dense-test}, we analyze the runtime as follows.
\begin{claim}
\claimlab{AC-dense-test-time}
The update time of line~\ref{line:AC-add-vertex} of \Cref{alg:clique-generation} induced by line~\ref{line:u-updates} (resp. line~\ref{line:v-updates}) is $O(\deg(u)\cdot \frac{\log{n}}{\eps^2})$ (resp. $O(\deg(v)\cdot \frac{\log{n}}{\eps^2})$).
\end{claim}
\begin{proof}
Note that we are essentially proving the runtime for \Cref{alg:AC-dense-test}. Assume w.log. that we only deal with $G[U]$. For each almost-clique $K \subseteq G[U]$ formed by the local decomposition algorithm, by the requirement of the algorithm, unless it returns ``fail'' and terminates the process, the number of vertices in $N(D)$ is at most $200\log{n}\cdot \card{K}$. Furthermore, for each vertex $u \in N(D)$, we need to check at most $100\cdot \frac{\log{n}}{\eps}$ neighbors. As such, the total running time for each almost-clique $K$ is at most $O(\frac{\log^2{n}}{\eps}\cdot \card{K})$. We can then sum up all vertices in $G[U]$, and get the overall running time of at most 
\begin{align*}
\sum_{K\subseteq G[U]} O\left(\frac{\log^2{n}}{\eps}\cdot \card{K}\right) =O\left(\frac{\log^2{n}}{\eps}\cdot \deg(v)\right),
\end{align*}
where the last step uses the disjointness of almost-cliques. This is as desired by the claim statement.
\myqed{\claimref{AC-dense-test-time}}
\end{proof}

By \claimref{AC-dense-test-time}, we conclude that the update time for lines \ref{line:u-updates} and \ref{line:v-updates} is at most $O(\deg(u)\cdot \frac{\log{n}}{\eps^2})$.
Finally, for the running time of line~\ref{line:sparse-vertex-removal}, it is easy to observe that each run of \Cref{alg:vertex-dense-test} only takes $O(\frac{\log^2{n}}{\eps^2})$ time since we sample at most $O(\frac{\log{n}}{\eps})$ vertices, and check $O(\frac{\log{n}}{\eps})$ neighbors for each of the vertices. Furthermore, note that we only check vertices in $U\cup \widetilde{U}$; and if $u$ invokes the procedure, we would need to check at most $\deg(u)\leq (1+\eps)\cdot \degT(u)$ vertices (resp. the same for $v$).

\end{proof}


Next, we show that the amortized update time of line~\ref{line:almost-clique-dismantle} is similarly bounded, although the argument takes another approach.
\begin{lemma}
\label{lem:update-time-ACs}
The amortized update time for each edge insertion/deletion update in line~\ref{line:almost-clique-dismantle} of \Cref{alg:dynamic-alg} is $O(\frac{1}{\eps^2})$.
\end{lemma}
\begin{proof}
For each almost-clique $K$, each time of execution of line~\ref{line:almost-clique-dismantle} of \Cref{alg:dynamic-alg} will take $O(\card{K})$ time. Therefore, it suffices to prove that each update happens after at least $\eps^2 \card{K}$ insertions/deletions. To see this, note that in our algorithm, the size of an almost-clique $K$ cannot increase by adding vertices to $K$ (the only way for the size of almost-cliques to increase is through the local SDDs in lines~\ref{line:u-updates} and \ref{line:v-updates}). As such, each run of line~\ref{line:almost-clique-dismantle} implies the removal of at least $\eps \card{K}$ \emph{vertices} $w\in K$ from $K$ by line~\ref{line:sparse-vertex-removal}.

Next, we note that for a vertex $w$ to be updated by line~\ref{line:sparse-vertex-removal}, $w$ has to be in $N(v)$ for some vertex $v$ such that $c_v \geq \frac{\eps}{10}\cdot \degT(v)$. Since each vertex could update at most $\deg(v)\leq (1+\frac{\eps}{10})\cdot \degT(v)$ vertices, we could show that
\begin{align*}
\text{number of edge updates} \geq \eps \cdot \text{number of vertices updated line~\ref{line:sparse-vertex-removal}}.
\end{align*}
Therefore, we could show that the number of edge updates is at least $\eps\cdot \eps\cdot \card{K}$ when line~\ref{line:almost-clique-dismantle} is invoked for almost-clique $K$. This leads to our desired lemma statement.
\end{proof}

\paragraph{Finalizing the proof of \Cref{lem:dynamic-alg-update-time}.} All the updates time are covered by lines~\ref{line:u-updates}, \ref{line:v-updates}, \ref{line:sparse-vertex-removal}, and \ref{line:almost-clique-dismantle}. By \Cref{lem:update-time-vertex}, the amortized update time for lines~\ref{line:u-updates}, \ref{line:v-updates}, and \ref{line:sparse-vertex-removal} is $O(\frac{\log^{2}{n}}{\eps^3})$. Furthermore, by \Cref{lem:update-time-ACs}, the amortized update time for each edge insertion/deletion is $O(\frac{1}{\eps^2})$. As such, the total amortized update time is at most $O(\frac{\log^{2}{n}}{\eps^3})$.

\subsubsection*{Approximation analysis} 
We now turn to the analysis of the approximation factor. Our goal is to show that the algorithm maintains an $O(\eps)$-sparse dense decomposition at any time, which, in turn, will lead to an $O(1)$-approximation of the optimal correlation clustering by the result of \cite{Assadi022}. The formal statement of the lemma is as follows.
\begin{lemma}
\label{lem:sparse-dense-decomp-maintain}
For any $\eps\leq 1/500$ and polynomial-bounded number of updates, with high probability, \Cref{alg:dynamic-alg} does \emph{not} output ``fail'', and it maintains a sparse-dense decomposition with the following properties at any point.
\begin{itemize}
\item Every sparse vertex is at least $\eps/8$-sparse.
\item Every almost-clique is at most $120\eps$-dense.
\end{itemize}
\end{lemma}
Note that combining \Cref{lem:sparse-dense-decomp-maintain} with \Cref{prop:sdd-approx} would straightforwardly lead to the desired $O(1)$-approximation as in \Cref{thm:main-alg}.
As such, our main task is to prove \Cref{lem:sparse-dense-decomp-maintain} in the rest of the analysis for approximation factors. For technical convenience, we assume w.log. that all vertices are of degree at least $100\cdot {\log{n}}/{\eps}$ in the analysis. All the results will go through when the degrees are less than $100\cdot {\log{n}}/{\eps}$ -- the estimations will be deterministic, and we could only get stronger statements.

\paragraph{The analysis of \Cref{lem:sparse-dense-decomp-maintain}.}
We now turn to the formal analysis of \Cref{lem:sparse-dense-decomp-maintain}.
Our strategy is to use an ``inductive'' argument to show that every time an update is invoked, we can always produce a sparse-dense decomposition such that every \emph{updated} sparse vertex is at least $\eps/2$-sparse and every \emph{updated} almost-clique is at most $2\eps$-dense. To this end, we define the \emph{updated} vertices as follows.
\begin{definition}
\label{def:updated-vertices}
At every time that updates happen in \Cref{alg:dynamic-alg} (i.e, the algorithm runs lines~\ref{line:u-updates}, \ref{line:v-updates}, \ref{line:sparse-vertex-removal}, and \ref{line:almost-clique-dismantle}), we say a vertex $v \in V$ is \emph{updated} if the counter $c_v$ is reset to $0$ and $\degT(v)\gets \deg(v)$ is executed. We say an almost-clique is an \emph{updated} almost-clique if there exists a vertex $v\in K$ such that $v$ is an updated vertex. We use $\Vup$ to denote all the updated vertices.
\end{definition}

The inductive statement we would prove for the rest of the analysis is as follows.
\begin{lemma}
\label{lem:sparse-dense-decomp-update}
For any $\eps\leq 1/500$ and polynomial-bounded number of updates, at every time that updates happen in \Cref{alg:dynamic-alg}, with high probability, the algorithm does \emph{not} return ``fail'', and there are
\begin{enumerate}[label=(\Roman*).]
\item All vertices $v\in V$ satisfies the properties as prescribed by \Cref{lem:sparse-dense-decomp-maintain}.
\item For all the \emph{updated vertices} as defined in \Cref{def:updated-vertices}, there are
\begin{itemize}
\item Every sparse vertex is at least $\eps/4$-sparse.
\item Every almost-clique is at most $20\eps$-dense.
\end{itemize}
\end{enumerate}
\end{lemma}

Observe that \Cref{lem:sparse-dense-decomp-maintain} simply follows from \Cref{lem:sparse-dense-decomp-update}. As such, our remaining task is to prove \Cref{lem:sparse-dense-decomp-update}. Our inductive argument on \emph{updated vertices} is as follows. 

\paragraph{The base case.} In the base case, we show that when the first edge insertion arrives, the properties prescribed by \Cref{lem:sparse-dense-decomp-update} hold. This is straightforward to see: let $(u,v)$ be the first edge inserted; the algorithm will run line~\ref{line:u-updates} and make $(u,v)$ an almost-clique. All other vertices in $V\setminus \{u,v\}$ are sparse, and the almost-clique is $0$-dense. 

\paragraph{Inductive steps.} Let us assume that the statement holds after $t-1$ algorithm updates (not to be confused with edge updates). Let $G_{t-1}$ be the graph right after the $(t-1)$-th algorithm update, and let $G_{t}$ be the graph of the $t$-th algorithm update. Furthermore, let $\SDD_{G_{t-1}, \eps}$ be the global sparse-dense decomposition (as vertex partition) before the $t$-th update happens. We show that \Cref{alg:dynamic-alg} that takes $\SDD_{G_{t-1}, \eps}$ and $G_t$ will produce a new $\SDD_{G_{t}, \eps}$ that satisfies the requirement of \Cref{lem:sparse-dense-decomp-update}. To this end, we first show that for all the vertices that are \emph{not} updated by the algorithm, the properties as prescribed by \Cref{lem:sparse-dense-decomp-update} (also \Cref{lem:sparse-dense-decomp-maintain}) will be satisfied. We prove the properties for almost-cliques first.

\begin{figure}
  \centering
  \begin{subfigure}[t]{.4\linewidth}
    \centering\includegraphics[width=\linewidth]{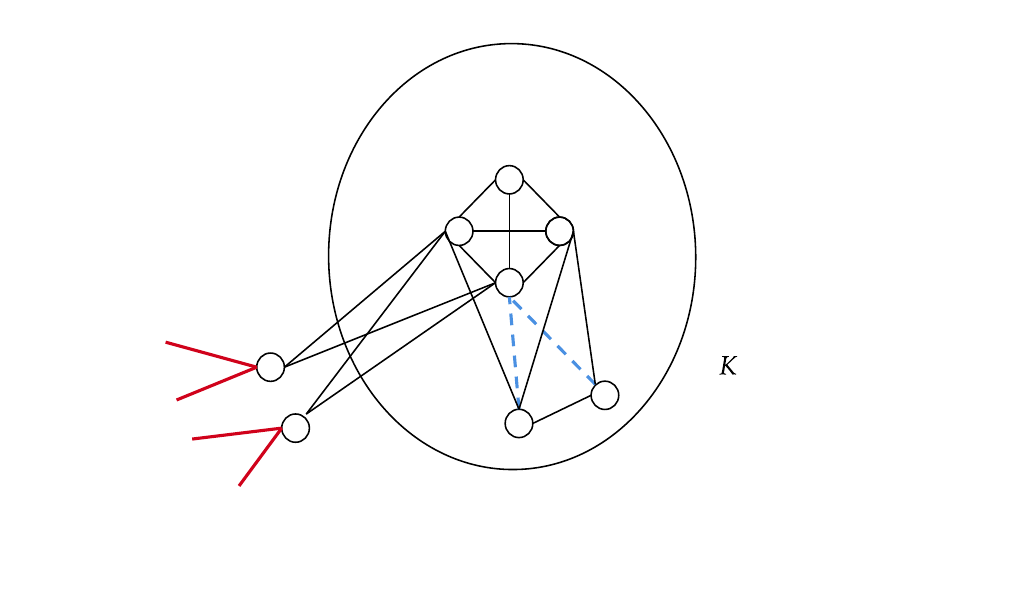}
    \caption{}
    \label{fig:without-update-AC}
  \end{subfigure}
  \begin{subfigure}[t]{.4\linewidth}
    \centering\includegraphics[width=\linewidth]{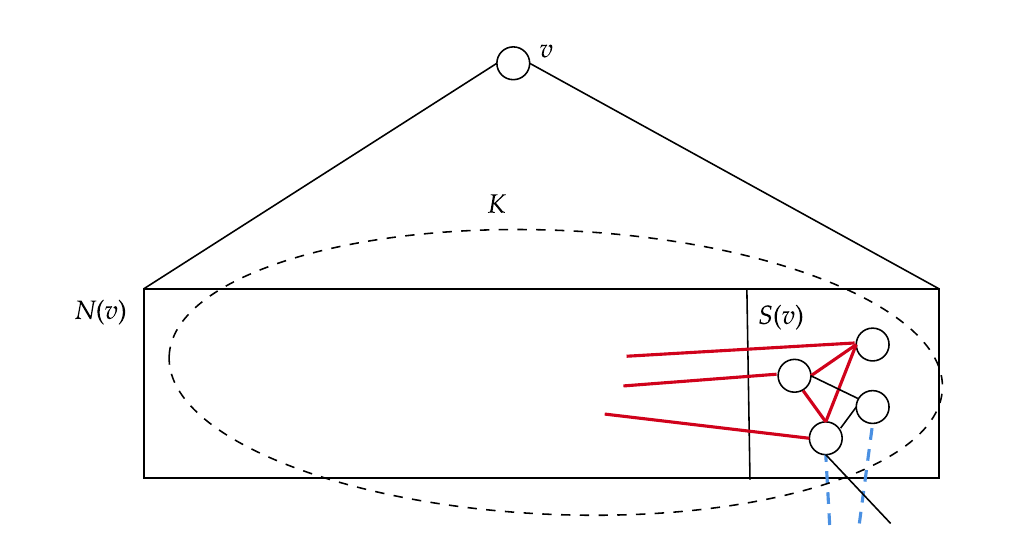}
    \caption{}
    \label{fig:without-update-sparse}
  \end{subfigure}
  \caption{An illustration of the analysis in \Cref{lem:not-update-AC-guarantee} and \Cref{lem:not-update-sparse-guarantee}. Red solid edges are the insertions and blue dotted edges are the deletions since the almost-clique and sparse vertices are formed. \Cref{fig:without-update-AC}: there could not be too many edge updates that make the almost-clique sparse since otherwise either the vertex itself becomes sparse, or the entire almost-clique is dismantled. \Cref{fig:without-update-sparse}: initially all vertices in $S(v)$ is sparse. If sufficiently many of them become not sparse, an update that leads to almost-clique $K$; and by \Cref{alg:AC-dense-test}, $v$ would have joined $K$.}
\end{figure}

\begin{lemma}
\label{lem:not-update-AC-guarantee}
With high probability, after the execution of the algorithm updates, the set of almost-clique vertices that are \emph{not updated} (i.e., vertices in $V\setminus \Vup$) satisfy the properties as specified by \Cref{lem:sparse-dense-decomp-maintain}.
\end{lemma}
\begin{proof}
For an almost-clique $K$ that has not been updated, the lines~\ref{line:u-updates}, \ref{line:v-updates}, and \ref{line:almost-clique-dismantle} should not have been executed on $K$. As such, let $k$ be the size of $k$ when the almost-clique is formed; by the conditions in line~\ref{line:almost-clique-dismantle}, $K$ has not lost more than $\eps$ fraction of the vertices. As such, the size of $K$ satisfies that $(1-\eps)k\leq \card{K}\leq k$. 
 We analyze the possible cases as follows.
 
\begin{enumerate}[label=\Alph*).]
\item If more than $20\eps k$ vertices $v\in K$ is initially at most $\frac{2}{5}\cdot \eps$-sparse when $K$ is formed. In this case, we note that for each $v\in V$, there is $(1-\eps/100)\cdot \degT(v)\leq \deg(v)\leq (1+\eps/100)\cdot \degT(v)$. Otherwise, by the induction hypothesis that $K$ is at most $20\eps$-dense, and since the vertex $v$ is at most $\eps/2$-sparse, the almost-clique will be updated during the execution of line~\ref{line:SDD-alg-invoke} of \Cref{alg:clique-generation}. We could therefore verify the properties of almost-clique vertices as follows.
\begin{itemize}
\item  The number of neighbors inside $K$. By the induction hypothesis, when the almost-clique is formed, for each vertex $v\in K$, there are at most $20\eps k$ non-neighbors for $v$ in $K$. At any point, for the vertices $v$ that are still in $K$, we have that $\deg(v)\geq (1-\eps/100)\degT(v)$ since otherwise an update must have happened. Therefore, the number of non-neighbors inside $K$ for vertex $v$ is at most
\begin{align*}
\card{K\setminus N(v)} &\leq 20\eps k + \eps/100\cdot \degT(v)\\
&\leq 20\eps k + \frac{\eps}{100(1-\eps/100)}\cdot \deg(v)\\
&\leq 20\eps \frac{\card{K}}{1-\eps}\cdot  + \frac{\eps}{100(1-\eps/100)}\cdot \deg(v)\\
&\leq  20\eps \frac{\deg(v)}{(1-\eps)\cdot (1-\eps/100)}\cdot  (1+20\eps) + \frac{\eps}{100(1-\eps/100)}\cdot \deg(v)\\
&\leq 40 \eps \deg(v) \tag{using $\eps\leq 1/500$},
\end{align*}
which is as desired by the lemma statement.
\item The number of neighbors outside $K$. By the induction hypothesis, when the almost-clique is formed, there are at most $20\eps k$ vertices going outside $K$. Again, since the almost-clique has not been updated, we have that $i).$ there are at most $\eps k$ vertices removed; and $ii).$ $\deg(v)\leq (1+\eps/100)\degT(v)$. As such, we have that
\begin{align*}
\card{V(v)\setminus K} &\leq 20\eps k + \eps/100\cdot \degT(v) + \eps k\\
&\leq 21\eps k + \frac{\eps}{100(1-\eps/100)}\cdot \deg(v)\\
&\leq 21\eps \frac{\card{K}}{1-\eps}\cdot  + \frac{\eps}{100(1-\eps/100)}\cdot \deg(v)\\
&\leq  21\eps \frac{\deg(v)}{(1-\eps)\cdot (1-\eps/100)}\cdot  (1+20\eps) + \frac{\eps}{100(1-\eps/100)}\cdot \deg(v)\\
&\leq 40 \eps \deg(v) \tag{using $\eps\leq 1/500$},
\end{align*}
which is as desired by the lemma statement.
\item The size of the almost-clique. By the induction hypothesis, by the time the almost-clique is formed, we have
\begin{align*}
(1-20\eps)\cdot \DeltaT(K) \leq k \leq (1+20\eps)\cdot \DeltaT(K),
\end{align*}
where we use $\DeltaT(K)$ to denote the maximum degree of vertex when $K$ is formed. 
Again, since the number of insertions has not triggered an update, there is 
\begin{align*}
(1-\eps/100)\cdot \DeltaT(K) \leq \Delta(K) \leq (1+\eps/100)\cdot \DeltaT(K).
\end{align*}
Therefore, we could prove the size bounds with
\begin{align*}
& \card{K} \geq (1-\eps)\cdot k \geq (1-\eps)\cdot (1-20\eps)\cdot \frac{\Delta(K)}{1+\eps/100} \geq (1-40\eps)\cdot \Delta(K)\\
& \card{K} \leq (1+\eps)\cdot k \leq (1+\eps)\cdot (1+20\eps)\cdot \frac{\Delta(K)}{1-\eps/100} \leq (1+40\eps)\cdot \Delta(K),
\end{align*}
where $\eps\leq 1/500$ was used in the above inequalities. This established the desired size bound for the almost clique.
\end{itemize}
In summary, for the case when $v\in K$ is initially at most $\frac{2}{5}\cdot \eps$-sparse when $K$ is formed, we have that the almost-clique is at most $40\eps$-dense.

\item  If at least $(1-20\eps)k$ vertices initially are at least $\frac{2}{5}\cdot \eps$-sparse in $K$ when $K$ is formed. In this case, if $v$ itself is at most $\frac{2}{5}\cdot \eps$-sparse, we could follow the analysis as in the above case. Otherwise, we claim that we have $(1-85\eps)\cdot k\leq \deg(v)\leq (1+85\eps)k$ for all $v\in K$. Note that by \Cref{lem:sparse-split-test}, since $v$ has \emph{not} been removed from $K$, it is at most $42\eps$-sparse. We prove the lower and upper bounds as follows.
\begin{itemize}
\item $\deg(v)\geq (1-85\eps) \cdot k$. Suppose for the purpose of contradiction that the property is not satisfied, i.e., $\deg(v)<(1-85\eps) \cdot k$. By the induction hypothesis, we have that at the time $K$ is formed, $v$ has at most $20\eps k$ neighbors outside $K$, and $20\eps k$ non-neighbors inside $K$. Furthermore, let $Z$ be the set of neighbors removed from a fixed $v\in K$, and let $S(v)$ be the initial set of sparse vertices of $v$. By the properties of almost-cliques, there must be $\card{N(v)\cap N(x)}\geq (1-22\eps)\cdot k$ for $x \in N(v)\setminus S(v)$ and $\card{N(v)\setminus S(v)} \geq (1-22\eps)\cdot k$. We now analyze some sub-cases as follows.
\begin{itemize}
\item If for less than $1/2$-fraction of the vertices in $z\in Z$, the adversary also remove edge $(x,z)$ for $x\in N(v)\setminus S(v)$. Then, for every all vertices $x \in N(v)\setminus S(v)$, we have that
\begin{align*}
\card{N(v)\sym N(x)} & \geq 32\eps\cdot k\geq 42\eps\cdot \deg(v),
\end{align*}
where the last inequality uses $\eps\leq 1/500$. Furthermore, we have that
\begin{align*}
\card{N(v)\setminus S(v)} & \geq (1-22\eps)\cdot k \geq 42\eps\cdot \deg(v),
\end{align*}
where the last inequality uses $\eps\leq 1/500$. This means the vertex $v$ is $42\eps$-sparse, which contradicts the fact that it is \emph{not} removed from $K$.
\item If for more than $1/2$-fraction of the vertices in $z\in Z$, the adversary also remove edge $(x,z)$ for $x\in N(v)\setminus S(v)$. Let this set of vertices be $Z'$, and not that the vertices in $Z'$ has become at least $42\eps$-sparse. Furthermore, we have that 
\begin{align*}
\card{Z'}\geq \frac{(85-20)\cdot \eps}{2} >2\eps,
\end{align*}
which contradicts the fact that $K$ has at most $2\eps$ of the vertices removed. 
\end{itemize}
Summarizing the above cases gives us the degree lower bound for $v$.

\item $\deg(v)\leq (1+85\eps) \cdot k$. Let $Z$ be the set of vertices in $N(v)\setminus K$. We first note that \emph{all} vertices $z\in Z$ should have that
\begin{itemize}
\item either $z$ is at least $\frac{2\eps}{5}$-sparse;
\item or at least $1/2$-fraction of the neighbors of $z$ are \emph{not} in $K$.
\end{itemize}
Again, this is using the fact that if both conditions are violated, then $z$ should have induced algorithm update on $K$ in line~\ref{line:SDD-alg-invoke} of \Cref{alg:clique-generation}. 
Furthermore, since we have that $v$ is initially at least $\frac{2}{5}\eps$ sparse, we have that $\card{Z}\geq \frac{\eps}{4}\cdot k$. 

Suppose for the purpose of contradiction that $\deg(v)> (1+85\eps) \cdot k$. By the induction hypothesis, there are at most $20\eps k$ non-neighbors for $v$ in $K$. Therefore, $v$ must have inserted at least $65\eps\cdot k$ vertices outside $K$. By \Cref{lem:dense-merge-test}, these vertices are again $\eps$-sparse. Together with the vertices in $Z$ that have to be sparse at all times, we could argue that vertex $v$ is at least $42\eps$-sparse, which contradicts the fact that it has not been removed from $K$.
\end{itemize}

Using the condition of $(1-85\eps)\cdot k\leq \deg(v)\leq (1+85\eps)k$, we now verify the properties of the almost-clique $K$ as follows.
\begin{itemize}
\item  The number of neighbors inside $K$. By the induction hypothesis, when the almost-clique is formed, for each vertex $v\in K$, there are at most $20\eps k$ non-neighbors for $v$ in $K$ and at most $20\eps$ neighbors for $v$ in $K$. 
Therefore, the number of non-neighbors inside $K$ for vertex $v$ is at most
\begin{align*}
\card{K\setminus N(v)} &\leq 20\eps k + 85 \eps k\\
&\leq 105\eps \cdot k\\
&\leq \frac{105\eps}{1-85\eps}\cdot \deg(v)\\
& \leq 120\eps \cdot \deg(v),    \tag{using $\eps\leq 1/500$}
\end{align*}
which is as desired by the lemma statement.

\item The number of neighbors outside $K$. By the induction hypothesis, when the almost-clique is formed, there are at most $20\eps k$ vertices going outside $K$. Again, since the almost-clique has not been updated, we have that there are at most $\eps k$ vertices removed. As such, we have that
\begin{align*}
\card{N(v)\setminus K} &\leq 20\eps k + 85 \eps k + \eps k\\
&\leq 106\eps \cdot k\\
&\leq \frac{106\eps}{1-85\eps}\cdot \deg(v)\\
& \leq 120\eps \cdot \deg(v),    \tag{using $\eps\leq 1/500$}
\end{align*}
which is as desired by the lemma statement.
\item The size of the almost-clique. By the induction hypothesis, by the time the almost-clique is formed, we have
\begin{align*}
(1-20\eps)\cdot \DeltaT(K) \leq k \leq (1+20\eps)\cdot \DeltaT(K),
\end{align*}
where we use $\DeltaT(K)$ to denote the maximum degree of vertex when $K$ is formed. 
Again, by the degree bound we have on the vertices $v\in K$, there is 
\begin{align*}
(1-85\eps)\cdot \DeltaT(K) \leq \Delta(K) \leq (1+85\eps)\cdot \DeltaT(K).
\end{align*}
Therefore, we could prove the size bounds with
\begin{align*}
& \card{K} \geq (1-\eps)\cdot k \geq (1-\eps)\cdot (1-20\eps)\cdot \frac{\Delta(K)}{1+85\eps} \geq (1-120\eps)\cdot \Delta(K)\\
& \card{K} \leq (1+\eps)\cdot k \leq (1+\eps)\cdot (1+20\eps)\cdot \frac{\Delta(K)}{1-85\eps} \leq (1+120\eps)\cdot \Delta(K),
\end{align*}
where $\eps\leq 1/500$ was used in the above inequalities. This established the desired size bound for the almost clique.
\end{itemize}
\end{enumerate}
\end{proof}

We note that although the analysis of \Cref{lem:not-update-AC-guarantee} is quite technical, it could be succinctly summarized as an illustration in \Cref{fig:without-update-AC}. We now turn to the case for sparse vertices.
\begin{lemma}
\label{lem:not-update-sparse-guarantee}
With high probability, after the execution of the algorithm updates, the set of sparse vertices (i.e., $\Vsparse$) that are \emph{not updated} (i.e., vertices in $V\setminus \Vup$) satisfy the properties as specified by \Cref{lem:sparse-dense-decomp-maintain}.
\end{lemma}
\begin{proof}
By the induction hypothesis, when a vertex $v\in \Vsparse$ was updated (i.e., newly-formed), the vertex must be at least $\eps/4$ sparse. Suppose for the purpose of contradiction that $v$ becomes less than $\eps/8$ sparse at some point. Let $S(v)$ be the set of vertices such that for $u\in S(v)$, $\card{N(v)\sym N(u)}\geq \frac{\eps}{4}\cdot \max\{\deg(u), \deg(v)\}$ when $v$ is first put into $\Vsparse$. Since we update $N(u)$ for each vertex $u\in V$ after the degree changes by an $\eps/100$ factor, there must exist a time when $v$ is sparse with a parameter of at most 
\[\frac{\eps}{8(1-\eps/100)}<2\eps/5,\]
and either $v$, a vertex $u\in N(v)$, or a two-hop neighbor of $v$ induces an update. We claim that at the time of the update, the vertex $v$ will join an almost clique $K$. To see this, note that for every vertex $u\in N(v)$ \emph{except} $2\eps/5\cdot \deg(v)$ vertices, there is 
\begin{align*}
\card{N(u)\sym N(v)}\leq \frac{2\eps}{5}\cdot \max\{\deg(u), \deg(v)\}.
\end{align*}
As such, during the execution of lines~\ref{line:u-updates} and \ref{line:v-updates}, none of these vertices will be classified to $\Vsparse$ by the guarantee in \Cref{prop:sdd-local-alg}. We now discuss the following cases
\begin{enumerate}[label=\alph*).]
\item If the update is induced by $v$ or $u\in N(v)$, then $v$ will join the almost-clique $K$.
\item If the update is induced by $w$ such that $w\in N(u)$, $u\in N(v)$, and $w\not\in N(v)$, then, initially after the execution of lines lines~\ref{line:u-updates} and \ref{line:v-updates}, $v$ would not join the almost-clique $K$. However, we claim that there is 
\begin{align*}
\card{N(v)\cap K} \geq (1-\eps/2) \cdot \deg(v) \geq \frac{1-2\eps/5}{1+\eps/2} \cdot \card{K} \geq (1-\eps)\cdot \card{K},
\end{align*}
where the first inequality uses the fact that only $\eps/2$ fraction of neighbors of $v$ has a large symmetric difference, and the second inequality uses the fact that the symmetric difference is at most $\frac{2\eps}{5}\cdot \max\{\deg(u), \deg(v)\}$. Therefore, during the execution of line~\ref{line:AC-add-vertex} in \Cref{alg:clique-generation}, and by the first bullet of \Cref{lem:dense-merge-test}, $v$ would have joined the almost-clique. 
\end{enumerate}
In both cases, $v$ gets updated and joins an almost-clique, which contradicts the assumption that $v$ is \emph{not} $\eps/8$-sparse (at most $\eps/8$-sparse). Therefore, $v$ must be at least $\eps/8$-sparse at all times.
\end{proof}

An illustration of the proof of \Cref{lem:not-update-sparse-guarantee} can be found in \Cref{fig:without-update-sparse}.
\Cref{lem:not-update-AC-guarantee} and \Cref{lem:not-update-sparse-guarantee} ensures that at the point we perform updates for the $t$-th algorithm update, the vertices that are not updated cannot become ``too bad'' since otherwise the updates would have occurred. 
We now show that during the execution of the update step, the algorithm does not return ``fail'' with high probability.
\begin{lemma}
\label{lem:alg-not-fail}
With high probability, after the execution of lines~\ref{line:u-updates}, \ref{line:v-updates}, \ref{line:sparse-vertex-removal}, and \ref{line:almost-clique-dismantle} of \Cref{alg:dynamic-alg}, the algorithm does \emph{not} return ``fail''.
\end{lemma}
\begin{proof}
Let $K'$ be an almost-clique formed by the execution of line~\ref{line:AC-forming}. Any vertex $w\in K'$ could have at most $\eps \card{K'}$ vertices outside. Furthermore, let $K$ be the surviving almost-clique after the execution of line~\ref{line:AC-checking}. Note that we have $(1-\eps)\cdot \card{K'} \leq \card{K} \leq \card{K'}$. As such, conditioning on the high-probability success of the algorithm in \Cref{prop:sdd-local-alg}, we have that
\begin{itemize}
\item The number of edges for $w$ going out of $K$ is at most $\frac{\eps}{2}\cdot \card{K'}\leq \frac{\eps}{2\cdot (1-\eps)}\cdot \card{K}\leq \eps \cdot \card{K}$; 
\item The number of non-neighbors for $w$ in $K$ could only \emph{decreases};
\item For a lower bound for the size of $K$, we have that
\begin{align*}
\card{K} & \geq (1-\eps)\cdot \card{K'} \\
& \geq (1-\eps)\cdot (1-\eps/2)\cdot \Delta(K') \tag{by the guarantees of $K'$}\\
& \geq (1-\eps)\cdot (1-\eps/2)\cdot \Delta(K) \tag{$\Delta(K)\geq \Delta(K')$} \\
& \geq (1-\eps/2)\cdot \Delta(K') \tag{holds true for $\eps\leq 1/2$}.
\end{align*}
\item For a upper bound for the size of $K$, we first note that for any vertex $w\in K'$, there is $\deg(w)\geq (1-\eps/2)\cdot \card{K'}$. As such, fix any $w\in K$, we have that
\begin{align*}
\card{K} & \leq \card{K'} \\
& \leq \frac{\deg(w)}{1-\eps/2} \tag{by the guarantees on $\deg(w)$}\\
& \leq \frac{\Delta(K')}{1-\eps/2} \\
& \leq (1+2\eps)\cdot \Delta(K') \tag{holds true for $\eps\leq 1/2$}.
\end{align*}
\end{itemize}
Therefore, we have that $K$ is at most $2\eps$-dense. By \Cref{lem:dense-merge-test}, this means the algorithm does not return ``fail'' conditioning only on the high-probability success of the algorithm in \Cref{prop:sdd-local-alg}.
\end{proof}

We now turn to the analysis of the \emph{updated} vertices.
We first deal with the almost-cliques; here, a concern is that some updated almost-cliques might be locally dense but globally sparse. As such, we show that with the procedures in lines~\ref{line:sparse-vertex-removal} and \ref{line:almost-clique-dismantle}, the almost-cliques will be indeed ``dense''.

\begin{figure}
    \centering\includegraphics[width=0.6\linewidth]{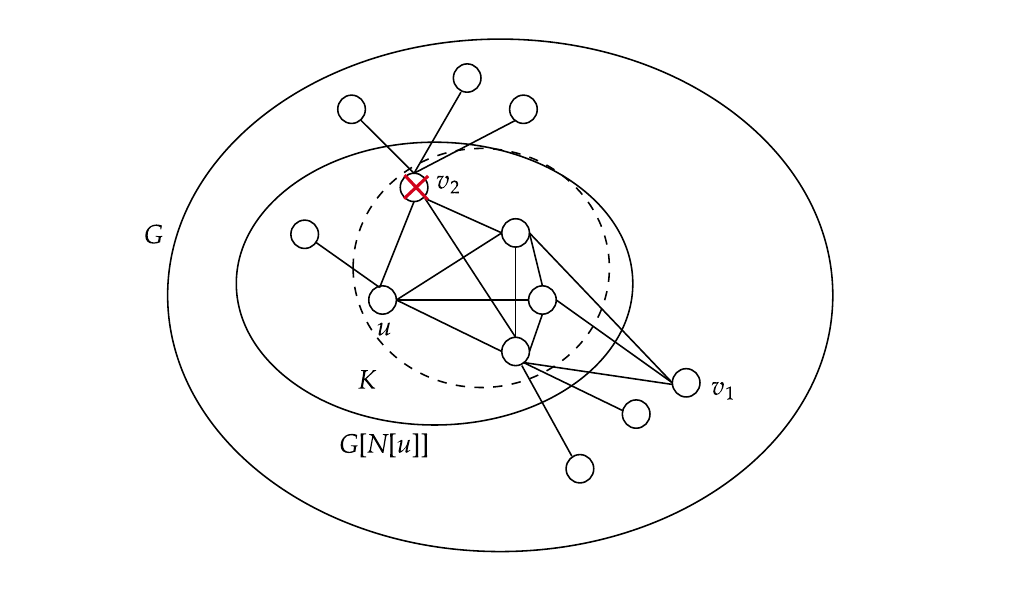}
  \caption{\label{fig:newly-formed-AC}An illustration of the role of \Cref{lem:AC-no-false}. Apart from the almost-clique that is formed locally, vertices like $v_1$ that are ``dense'' w.r.t. $K$ will be added by \Cref{alg:AC-dense-test} (see also~\Cref{fig:add-v-to-ac}). Furthermore, vertices like $v_2$ will be removed if it is sparse (see also~\Cref{fig:local-AC-check}).}
\end{figure}

\begin{lemma}
\label{lem:AC-no-false}
With high probability, after the execution of lines~\ref{line:u-updates}, \ref{line:v-updates}, \ref{line:sparse-vertex-removal}, and \ref{line:almost-clique-dismantle} of \Cref{alg:dynamic-alg}, all \emph{updated} almost-cliques (defined in \Cref{def:updated-vertices}) are at most $20\eps$-dense.
\end{lemma}
\begin{proof}
The only possible almost-cliques that are \emph{updated} are the ones formed by the invocation of line~\ref{line:SDD-alg-invoke} of \Cref{alg:clique-generation} by lines~\ref{line:u-updates} and \ref{line:v-updates}. As such, we analyze any fixed vertex $w\in K$ for some almost-clique $K$, and show that the properties as prescribed by \Cref{def:sdd} hold. For the sake of simplicity, we assume w.log. that only line~\ref{line:u-updates} (updates incurred by $u$) in executed since the analysis for line~\ref{line:v-updates} being invoked is essentially the same. 

We let $k$ be the size of the almost-clique that is returned by line~\ref{line:AC-forming} and $k'$ be the size of the al,ost-clique after the execution of line~\ref{line:AC-checking}. These quantities are defined to make it easier to control the quantities in the argument.

Let $N_{N[u]}(w)=N(w)\cap N[U]$ be the set of neighbors of $w$ that are in $N[u]$. We first argue that if $w$ is added to the almost-clique by line~\ref{line:SDD-alg-invoke} or line~\ref{line:AC-add-vertex} of \Cref{alg:clique-generation}, then $w$ must be ``locally'' dense. More formally, we show that
\begin{claim}
\claimlab{add-ac-vertex-dense}
Let $w$ be a vertex added to an almost-clique $K$ by line~\ref{line:SDD-alg-invoke} or line~\ref{line:AC-add-vertex}. Then, with high probability, there must be
\begin{itemize}
\item $w$ has at most $8\eps \card{K}$ non-neighbors inside $K$; and
\item $(1-2\eps)\cdot \card{N_{N[u]}(w)} \leq \card{K} \leq (1+10\eps)\cdot \card{N_{N[u]}(w)}$.
\end{itemize}
\end{claim}
\begin{proof}
If $w$ is added by line~\ref{line:SDD-alg-invoke}, then by the guarantee of \Cref{prop:sdd-local-alg}, we have that $w$ has at most $2\eps k$ non-neighbors inside $K$. Furthermore, we note that since $w$ has at most $\eps k$ non-neighbors inside $K$, there is $\card{N_{N[u]}(w)}\geq (1-\eps)\cdot k$, which implies $k\leq (1+2\eps)\cdot \card{N_{N[u]}(w)}$ for any $\eps\leq \frac{1}{2}$. Therefore, we establish that $(1-\eps)\cdot \card{N_{N[u]}(w)}\leq k\leq (1+2\eps)\cdot \card{N_{N[u]}(w)}$. Furthermore, since we remove at most $\eps k$ vertices during the process, we have that $(1-\eps)\cdot k\leq k'\leq k$. 

We now proceed to the vertices added by line~\ref{line:AC-add-vertex}.
If $w$ is added by line~\ref{line:AC-add-vertex}, there are two concerns we need to handle: we need to verify $w$ itself indeed satisfies the conditions in \claimref{add-ac-vertex-dense}, and we need to prove that the newly-added vertices do \emph{not} lead to the break of conditions for any other vertex $x$ added by line~\ref{line:SDD-alg-invoke}. For a single vertex $w$, we can directly use the second bullet of \Cref{lem:dense-merge-test} to show that $w$ has at most $4\eps k'$ non-neighbors inside $K$. Furthermore, in line~\ref{line:AC-forming} of \Cref{alg:clique-generation}, every vertex $v\in \card{K}$ has at most $3\eps k$ vertices \emph{outside} $\card{K}$, and at most $\frac{3\eps}{1-\eps}\cdot k$ vertices \emph{outside} $\card{K}$ after the execution of line~\ref{line:AC-checking}. Therefore, the total number of edges going out of $K$ is at most $\frac{3\eps}{1-\eps}\cdot k^{2}$.

Note that by the second bullet of \Cref{lem:dense-merge-test}, every vertex being added to $K$ has to have at least $(1-4\eps)\cdot k'$ vertices in the almost-clique. As such, the number of vertices added by line~\ref{line:AC-add-vertex} is at most
\begin{align*}
\frac{3\eps}{1-\eps}\cdot k^{2}\cdot \frac{1}{(1-4\eps)\cdot k'}\leq 4\eps\cdot k,
\end{align*}
where the last inequality follows for every $\eps\leq \frac{1}{50}$. Therefore, inside the updated $\card{K}$, each vertex could have at most $3\eps \cdot k + 4\eps \cdot k=7\eps\cdot k$ non-neighbors in the almost-clique.

Finally, note that during the process of lines~\ref{line:AC-checking} and \ref{line:AC-add-vertex}, the size of the almost-clique could become at most $(1-\eps)$ time smaller and $(1+4\eps)$ times larger. As such, we have that
\begin{align*}
(1-\eps)\cdot k \leq \card{K} \leq (1+4\eps)\cdot k' \leq  (1+4\eps)\cdot  (1+2\eps)\cdot k.
\end{align*}
Therefore, we can summarize the above inequalities, and obtain that for each $w$, the number of non-neighbors of $w$ inside $K$ is at most 
\[7\eps \cdot k\leq \frac{7}{1-\eps}\cdot \eps \cdot\card{K}\leq 8\eps\cdot \card{K},\]
where the last inequality holds for every $\eps\leq \frac{1}{500}$.
Furthermore, for the size bound of $K$, we can similarly obtain that
\begin{align*}
& \card{K} \geq (1-\eps)\cdot (1-\eps)\cdot\card{N_{N[u]}(w)} \geq (1-2\eps)\cdot \card{N_{N[u]}(w)}\\
& \card{K} \leq (1+2\eps)^{2}\cdot (1+4\eps)\cdot \card{N_{N[u]}(w)} \leq (1+10\eps)\cdot \card{N_{N[u]}(w)},
\end{align*}
where the above inequalities hold for $\eps\leq 1/500$. This gives the desired statement of \claimref{add-ac-vertex-dense}.
\myqed{\claimref{add-ac-vertex-dense}}
\end{proof}

By \claimref{add-ac-vertex-dense}, the only concern at this point is that vertex $w\in K$ might have too many neighbors outside $K$ (and in $V\setminus N[u]$). We exactly handle this by using running line~\ref{line:sparse-vertex-removal} of \Cref{alg:dynamic-alg}. Concretely, suppose vertex $w\in K$ has more than $8\eps$ neighbors \emph{outside} $K$, we argue that $w$ must be at least $2\eps$ sparse. To see this, let us denote $K^{\text{formed}}$ and $K^{\text{added}}$ as the set of vertices added to $K$ by the end of line~\ref{line:AC-checking} of \Cref{alg:clique-generation} and by \Cref{alg:AC-dense-test}, respectively. We now show that none of the categories could have too many vertices outside $K$.
\begin{itemize}
\item If $w\in K^{\text{formed}}$, then by the condition of line~\ref{line:AC-forming} and \ref{line:AC-checking} of \Cref{alg:clique-generation}, there is $\deg(w)\leq (1+2\eps)\cdot k$. Furthermore, by \claimref{add-ac-vertex-dense}, there are at most $7\eps/(1-\eps)\cdot k$ non-neighbors inside $\card{K}$. As such, the number of non-neighbors of $w$ outside $K$ is at most
\begin{align*}
(1+2\eps)\cdot k - \paren{\card{K}-7\eps/(1-\eps)\cdot k} \leq (1+10\eps)\cdot k - (1-\eps)\cdot k \leq 12\eps \cdot \card{K},
\end{align*}
where the last inequality uses the relationship between $\card{K}$ and $k$.
\item Similarly, if $w\in K^{\text{added}}$, then note that by the requirement of \Cref{alg:AC-dense-test}, the degree of $w$ is at most $(1+2\eps)\cdot k'\leq (1+2\eps)^{2}(1+4\eps)\cdot k$. Therefore, the number of neighbors going \emph{out} of $K$ from $w$ can be at most 
\begin{align*}
(1+2\eps)^{2}(1+4\eps)\cdot k - \paren{\card{K}-7\eps/(1-\eps)\cdot k} \leq (1+18\eps)\cdot k - (1-\eps)\cdot k \leq 20\eps \cdot \card{K},
\end{align*}
where the inequalities also follow from $\eps\leq \frac{1}{500}$.
\end{itemize}

Finally, during the execution of line~\ref{line:sparse-vertex-removal}, no almost-clique will be \emph{updated} (i.e., only vertex removal would happen to the almost-clique). As such, we reach our conclusion that the almost-clique has to be at most $20\eps$-dense.
\end{proof}

Note that our analysis has \emph{not} finished here: 
we need to show that for all \emph{updated} vertices, the procedure in lines~\ref{line:sparse-vertex-removal} and \ref{line:almost-clique-dismantle} does \emph{not} introduce to $\Vsparse$ with vertices that are \emph{not} sparse.
We now handle the concern with the following lemma.
\begin{figure}
  \centering
  \begin{subfigure}[t]{.43\linewidth}
    \centering\includegraphics[width=\linewidth]{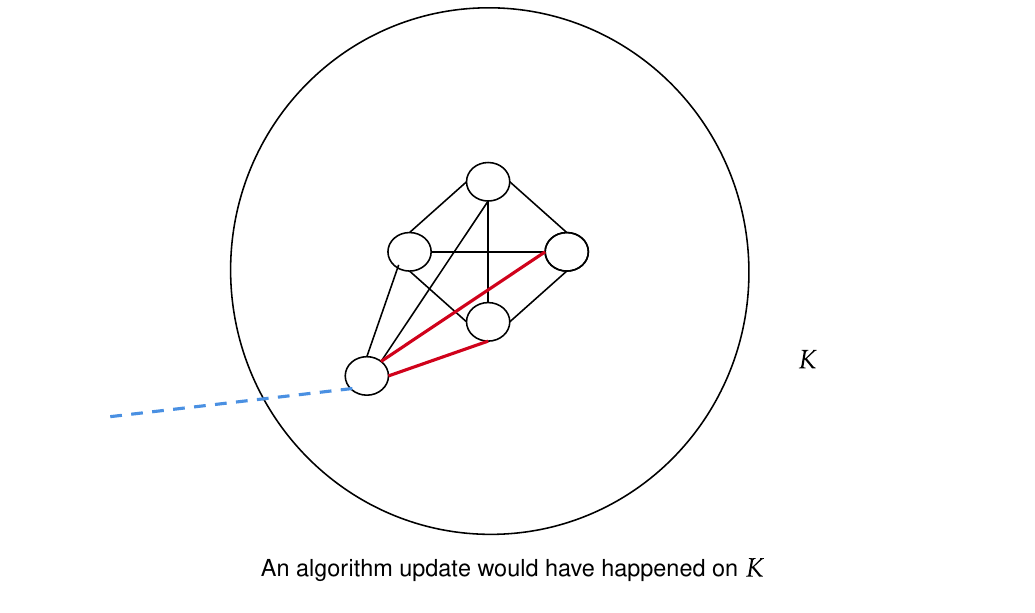}
    \caption{}
    \label{fig:sparse-vertex-update-case-a}
  \end{subfigure}
  \begin{subfigure}[t]{.43\linewidth}
    \centering\includegraphics[width=\linewidth]{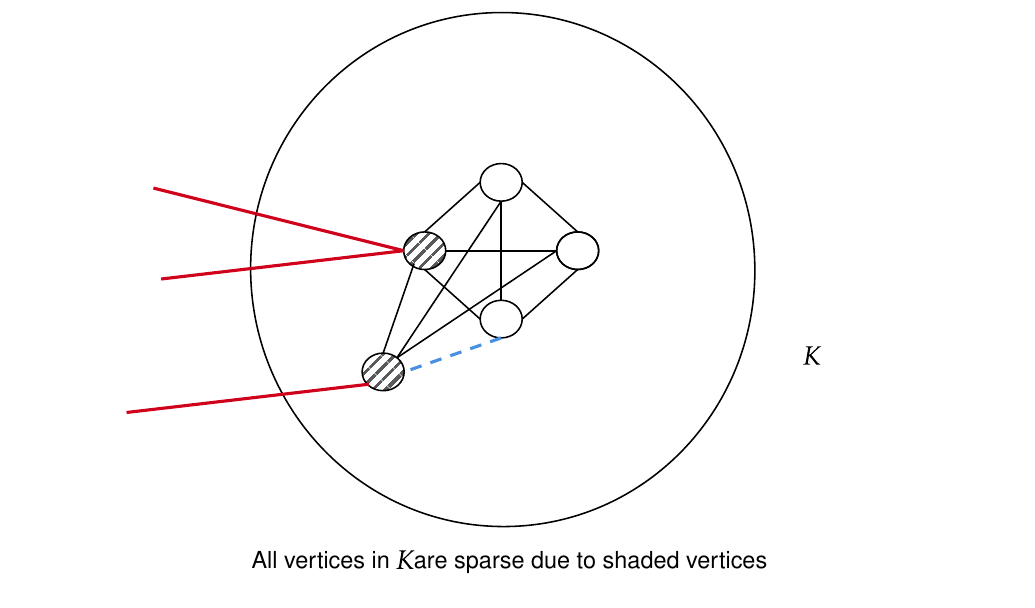}
    \caption{}
    \label{fig:sparse-vertex-update-case-b}
  \end{subfigure}
  \caption{An illustration of the case analysis in \Cref{lem:sparse-no-false}. \Cref{fig:sparse-vertex-update-case-a}: if a vertex $w\in K$ is initially somehow sparse, then $w$ could not be made very dense since otherwise, an algorithm update would happen. \Cref{fig:sparse-vertex-update-case-b}: if a vertex $w\in K$ is initially very dense, then the additional insertions that makes other vertices in $K$ sparse will also make $w$ sparse.}
\end{figure}
\begin{lemma}
\label{lem:sparse-no-false}
With high probability, after the execution of lines~\ref{line:u-updates}, \ref{line:v-updates}, \ref{line:sparse-vertex-removal}, and \ref{line:almost-clique-dismantle} of \Cref{alg:dynamic-alg}, all \emph{updated vertices} in $\Vsparse$ are at least $\frac{\eps}{4}$-sparse.
\end{lemma}
\begin{proof}
We again assume w.log. that only the updates on $u$ (line~\ref{line:u-updates}, \ref{line:sparse-vertex-removal}, and \ref{line:almost-clique-dismantle}) are executed. Note that a \emph{updated} vertex $w$ could be added to $\Vsparse$ by lines~\ref{line:sparse-vertex-removal} and \ref{line:almost-clique-dismantle}. If the vertex is added by line~\ref{line:sparse-vertex-removal}, then we can apply the (contra-positive of the) second bullet of \Cref{lem:sparse-split-test} to show that the vertex has to be at least $38\eps$-sparse, which implies that the vertex is at least $2\eps$-sparse.

On the other hand, if a vertex $w$ is added to $\Vsparse$ by line~\ref{line:almost-clique-dismantle}, we show that $w$ is still at least $\eps$-sparse. Note that in this case, there has to exist an almost-clique $K$ such that $w\in K$ before it becomes a sparse vertex. Let $K^{\text{sparse}}$ be the set of vertices that is removed from $K$ since it is formed, and let $k$ be the size of the almost-clique $K$ before the removal of any vertex. We first claim that in the graph $G_{t}$ (the graph of the current update step), the following properties are true.
\begin{claim}
\claimlab{removed-vertex-sparse}
By the time a vertex $v$ joins $K^{\text{sparse}}$, it must satisfy the following properties:
\begin{itemize}
\item either $v$ has at least $27\eps k$ neighbors \emph{outside} $K$;
\item or there $v$ has at least $27\eps k$ non-neighbors inside $K$.
\end{itemize}
\end{claim}
\begin{proof}
The proof of the above statement is as follows. Note that there are two ways for $v$ to be removed from $K$: either it becomes a member of another almost-clique, or it becomes a sparse vertex by line~\ref{line:sparse-vertex-removal} (\Cref{alg:vertex-dense-test}). In the former case, let $K'$ be the new almost-clique to which $v$ belongs, and all the edges to $K\setminus K^{\text{sparse}}$ would be considered edge  \emph{going out of} $K'$. As such, the number of edges for $v$ to be \emph{outside} $K$ should be at least
\begin{align*}
(1-\eps)k\cdot \frac{1-40\eps}{40\eps} \geq 26\eps k,
\end{align*}
where the last inequality holds for $\eps\leq 1/500$. This is as desired by the statement of \claimref{removed-vertex-sparse}.

On the other hand, if $v$ is removed from $K$ by line~\ref{line:sparse-vertex-removal} (\Cref{alg:vertex-dense-test}), then $v$ has to be $38\eps$-sparse the moment it is removed from $K$ (by \Cref{lem:sparse-split-test}). We claim that at this point, there must be
\begin{itemize}
\item either $v$ has at least $27\eps k$ neighbors \emph{outside} $K$;
\item or there $v$ has at least $27\eps k$ non-neighbors inside $K$.
\end{itemize}
To see this, suppose for the purpose of contradiction the statement is not true. Then, we have that 
\begin{itemize}
\item The degree of $v$ is comparable to $k$, i.e., $(1-27\eps)\cdot k \leq \deg(v) \leq (1+27\eps)\cdot k$; and
\item The intersection between $N(v)$ and $K$ is significant, i.e., $\card{N(v)\cap K} \geq (1-27\eps) \cdot k$.
\end{itemize}
Furthermore, since we have that $\card{K}\leq k \leq \frac{\card{K}}{1-\eps}$, we have that for all vertices $w\in K$, there is
\begin{align*}
\card{N(v)\sym N(w)}\leq 27\eps\cdot k \leq 38\cdot \eps\cdot \max\{\deg(v), \deg(w)\},.
\end{align*}
where the last inequality works for $\eps\leq 1/500$.
Furthermore, by the above calculation, we have that $\card{N(v)\setminus N(w)}\leq 26\eps k \leq 38\cdot \deg(v)$. As such, the vertex cannot be $38\cdot \eps$-sparse, which forms a contradiction. As such, one of the conditions in the claim statement has to be met.\myqed{\claimref{removed-vertex-sparse}}
\end{proof}

For vertices in $K\setminus K^{\text{sparse}}$ that are added to $\Vsparse$ by line~\ref{line:almost-clique-dismantle}, we analyze the cases as follows.
\begin{enumerate}[label=\roman*).]
\item If $w$ is initially at least $\frac{2}{5}\cdot \eps$ sparse. In this case, claim that there must exist a subset of vertices $S(w)\subseteq N(w)$, such that $\card{S(w)}\geq \frac{\eps}{4}$, and for every vertex $x\in S(w)$, there is
\begin{align*}
\card{N(w)\sym N(z)}\geq \frac{\eps}{4}\cdot \max\{\deg(w), \deg(z)\}.
\end{align*}
The claim simply follows from the guarantee of \Cref{prop:sdd-local-alg}. More concretely, suppose the condition is not true. Since $w$ starts with being at least $\frac{2}{5}\cdot \eps$ sparse, and we make updates every $\frac{\eps}{100}\cdot \degT(v)$ steps, we have that
\begin{align*}
\frac{\eps}{4(1-\eps/100)}\leq \frac{2}{5}\cdot \eps<\frac{\eps}{2}.
\end{align*}
Therefore, by the guarantees in \Cref{prop:sdd-local-alg}, the almost-clique must have been updated for $w$, which formed a contradiction with the necessary condition to enter line~\ref{line:almost-clique-dismantle}.
\item If $w$ is \emph{not} initially at least $\frac{2}{5}\cdot \eps$ sparse. In this case, note that we have
\begin{align*}
\frac{2\eps}{5(1-\eps/100)}<\eps/2
\end{align*}
for $\eps\leq 1/500$, which means $v$ could not have entered lines~\ref{line:u-updates} and \ref{line:v-updates} since otherwise it would also update the almost-clique. As such, there are no updates on $v\in K^{\text{sparse}}$ between the time when $v$ becomes a sparse vertex and the current update (i.e., no execution of lines~\ref{line:u-updates} and \ref{line:v-updates} induced by $v$). As such, the guarantees by \claimref{removed-vertex-sparse} remain valid.

By the induction hypothesis, when $K$ is formed, the vertices in $K$ are at most $20\eps$-dense. Furthermore, for every $w \in K\setminus K^{\text{sparse}}$, since $w$ has not been updated, we have that
\begin{align*}
& \deg(w) \geq (1-\eps/100)\cdot \frac{\card{K}}{1+20\eps} \geq (1-23\eps)\cdot k; \\
& \deg(w) \leq (1+\eps/100)\cdot \frac{\card{K}}{1-20\eps} \leq (1+23\eps)\cdot k.
\end{align*}
With the above conditions, we claim that for any vertex $v\in K^{\text{sparse}}$ and $w \in K\setminus K^{\text{sparse}}$, there is
\begin{align*}
\card{N(w)\sym N(v)}\geq \frac{\eps}{2}\cdot \max\{\deg(w), \deg(v)\}.
\end{align*}
To prove the above inequality, we perform a case analysis as follows. We assume w.log. that $\deg(v)\geq \deg(w)$ since the other case follows with the same analysis. We now analyze the following sub-cases.
\begin{enumerate}[label=\alph*).]
\item If $\deg(v)-\deg(w)\geq \eps/2\cdot \deg(v)$. In this case, we trivially have $\card{N(u)\sym N(v)}\geq \eps/2\cdot \max\{\deg(w), \deg(v)\}$.
\item If $0 \leq \deg(v)-\deg(w)\leq \eps/2\cdot \deg(v)$. In this case, by moving the terms around, we have that
\begin{align*}
\deg(w)\geq (1-\eps/2)\cdot \max\{\deg(w), \deg(v)\}.
\end{align*}
Therefore, we can apply \claimref{removed-vertex-sparse} and argue that 
\begin{align*}
\card{N(w)\sym N(v)} & \geq 26\eps k \\
& \geq \frac{26\eps}{1+23\eps}\cdot \deg(w) \\
& \geq \frac{26\eps}{1+23\eps}\cdot (1-\eps/2)\cdot \max\{\deg(w), \deg(v)\} \\
& \geq  \frac{\eps}{2}\cdot \max\{\deg(w), \deg(v)\},
\end{align*}
where the last inequality holds for $\eps\leq 1/500$.
\end{enumerate}
Furthermore, we note that the number of removed vertices in $K^{\text{sparse}}$ is at least
\begin{align*}
\card{K^{\text{sparse}}}\geq \eps k \geq \frac{\eps}{2}\cdot \max\{\deg(w), \deg(v)\},
\end{align*}
where the last inequality again uses $\eps\leq 1/500$.
Therefore, we conclude that $w$ is at least $\eps/2$-sparse.
\end{enumerate}
Summarizing the above cases gives us to desired lemma statement.
\myqed{\Cref{lem:sparse-no-false}}
\end{proof}
An illustration of the case analysis for \Cref{lem:sparse-no-false} can be found in \Cref{fig:sparse-vertex-update-case-a} and \Cref{fig:sparse-vertex-update-case-b}.


\paragraph{Wrapping up the induction step to prove \Cref{lem:sparse-dense-decomp-update}.} By \Cref{lem:AC-no-false} and \Cref{lem:sparse-no-false}, the newly-formed almost-cliques and sparse vertices satisfy the properties as prescribed by \Cref{lem:sparse-dense-decomp-update}. Furthermore, in \Cref{lem:not-update-AC-guarantee} and \Cref{lem:not-update-sparse-guarantee}, we show that the sparse vertices and almost-cliques that are \emph{not} updated remain in the range of the parameters. As such, we conclude that the properties are followed in the $t$-th algorithm update step, which concludes the inductive proof.

To see why our algorithm is adversarial-robust, note that our randomness in each step of the updates is \emph{independent} of the realization of past randomness. As such, our algorithm works no matter how the insertions and deletions are carried out.

\paragraph{Proof of \Cref{thm:main-alg}.} We simply run the algorithm of \Cref{lem:sparse-dense-decomp-maintain} with the adjacency list of $G^{+}=(V, E^{+})$ as the input graph. By \Cref{lem:sparse-dense-decomp-maintain}, we maintain a sparse-dense decomposition with $\eps$-sparse vertices and $\eps'$-dense almost-cliques, where both $\eps$ and $\eps'$ are constants. Therefore, we could maintain $O(1)$-approximation for the clustering at any time point.

%% file: subroutines.tex
\subsection{Technical Subroutines}
\label{sec:tech-subroutines}
Before presenting our main algorithm, we present some intermediate algorithms as technical subroutines. Conceptually, these algorithms use known ideas and follow from the previous algorithms for sparse-dense decomposition. Nevertheless, we remark that they are important for the purpose of our main algorithm, and their analysis is fairly technical and involved.

\paragraph{An algorithm that merges vertices to almost-cliques.}
We start by introducing a subroutine that merges ``candidate'' vertices to almost-cliques.

\begin{Algorithm}
\label{alg:AC-dense-test}
\textbf{An algorithm that adds vertices to an almost-clique.}\\
\textbf{Input:} A graph $G=(V,E)$; an almost-clique $K$ from some sparse-dense decomposition $\SDD_{G,\eps}$.
\smallskip
\begin{enumerate}
\item Sample a set $D(v)$ of $\min\{100\cdot\frac{\log{n}}{\eps}, \deg(v)\}$ vertices from $K$ uniformly at random.
\item Let $N(D)$ be the set of \emph{neighbors} of $D$, i.e., $N(D)=\cup_{u\in D(v)}N(u)$. If $\card{N(D)}\geq 200\log{n}\cdot \card{K}$, terminate and return ``fail''.
\item Sample a set of vertices $T$ of $\min\{100\cdot\frac{\log{n}}{\eps}, \deg(v)\}$ vertices from $K$ uniformly at random.
\item For vertices $u \in N(D)$ in parallel, perform the following procedure:
\begin{enumerate}
\item If $u$ has at least $(1-2\eps)\cdot \card{T}$ neighbors in $T$ \emph{and} $(1-2\eps)\cdot \card{K}\leq \deg(u)\leq (1+2\eps)\cdot \card{K}$, add $u$ to $K$.
\item Otherwise, keep $u$ as its current assignment (sparse vertex or in another almost-clique). 
\end{enumerate}
\item Return the updated decomposition of sparse vertices and almost-cliques.
\end{enumerate}
\end{Algorithm}

Roughly speaking, \Cref{alg:AC-dense-test} takes an almost-clique $K$, and adds vertices that were potentially ``missed'' from the almost-clique. The subroutine is important to keep dense vertices in almost-cliques: if $u$ induces an algorithm update, there might be a few vertices that belong to an induced almost-clique $K$ but \emph{not} a neighbor of $u$. The key observation here is that there could be only a \emph{small} number of such vertices, and they could be added to $K$ by \Cref{alg:AC-dense-test}. An illustration of the scenario and the guarantee of \Cref{alg:AC-dense-test} is shown as \Cref{fig:add-v-to-ac}.
\begin{figure}
    \centering\includegraphics[width=0.6\linewidth]{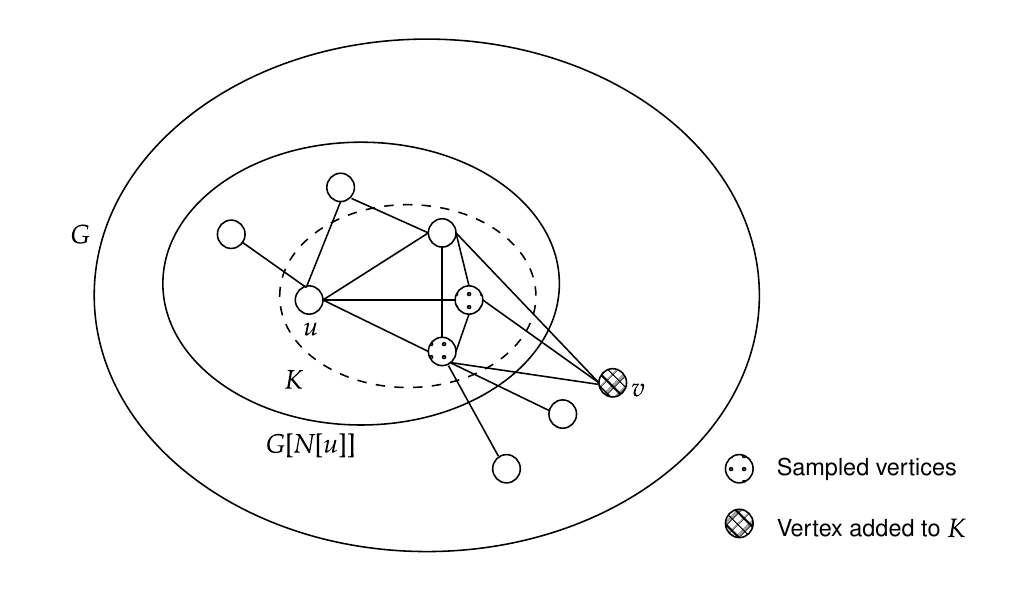}
  \caption{\label{fig:add-v-to-ac}An illustration of the role of \Cref{alg:AC-dense-test} and \Cref{lem:dense-merge-test}. The shaded vertex $v$ should be added to $K$; however, local updates cannot capture this due to the fact that $v$ is not a neighbor of $u$. \Cref{alg:AC-dense-test} samples the dotted vertices, and recognize their common neighbor, which are the vertices should be added to $K$.}
\end{figure}

The following technical lemma establishes the properties for \Cref{alg:AC-dense-test}. 
\begin{lemma}
\label{lem:dense-merge-test}
Let $K$ be an almost-clique that is at most $2\eps$-dense, and let $w$ be any vertex such that $(1-2\eps)\cdot \card{K}\leq \deg(w)\leq (1+2\eps)\cdot \card{K}$. The algorithm does \emph{not} return ``fail'', and with high probability, the following statements are true.
\begin{itemize}
\item If $\card{N(w)\cap K} \geq (1-\eps)\cdot \card{K}$, the vertex will be added to $K$.
\item If $\card{K \setminus N(w)} \geq 4\eps\cdot \card{K}$, the vertex will \emph{not} be added to $K$.
\end{itemize}
\end{lemma}
\begin{proof}
We first show that the algorithm does \emph{not} return ``fail''. Note that since $K$ is at most $2\eps$-dense, there are at most $2\eps\card{K}$ edges going out of $K$. As such, since we sample at most $\frac{100\log{n}}{\eps}$ vertices to form $D$, the total number of vertices in $N(D)$ is at most $2\eps\cdot \card{K}\cdot \frac{100\log{n}}{\eps} = 200\log{n}\cdot \card{K}$. This means the algorithm will \emph{not} return ``fail''.

We prove the statements in the two bullet points in order.
\begin{itemize}
\item For the first bullet, we first show that with high probability, the vertex $w$ will be added to $N(D)$ of \Cref{alg:AC-dense-test}. Since $\card{N(w)\cap K} \geq (1-\eps)\cdot \card{K}$, for a random vertex $v\in K$, we have that
\begin{align*}
\Pr\paren{w \not\in N(v)}\leq \eps.
\end{align*}
As such, the probability for $w$ not to be in \emph{any} $N(v)$ for $100\log{n}/\eps$ independent vertices is at most
\begin{align*}
\Pr\paren{w \not\in N(v) \text{ for all $v\in D(v)$}}\leq \eps^{100\log{n}/\eps} \leq \frac{1}{n^{20}}.
\end{align*}
Let us condition on the above high-probability event. Now that $w$ is sampled, we need to check that $w$ would be added to $K$. Indeed, during the sampling process of $T$, let $X_v$ be the indicator random variable for $v\not\in N(w)$, and let $X=\sum_{v\in T}X_v$ be the random variable for the \emph{total} number of non-neighbors of $w$ in $T$. We have that
\begin{align*}
\expect{X}&= \expect{\sum_{v\in T} X_v}\\
&= \sum_{v\in T}\expect{X_v} \tag{linearity of expectation}\\
&= \sum_{v \in T} \Pr\paren{v \not\in N(w)} \tag{by the definition of indicator random variable}\\
&\leq  100\cdot \frac{\log{n}}{\eps}\cdot \eps = 100 \cdot \log{n}.
\end{align*}
Note that for $\card{N(w)\cap T}\leq (1-2\eps)\card{T}$, there has to be at least $2\eps\cdot 100 \cdot \frac{\log{n}}{\eps} =200\log{n}$ vertices $v\in V$ such that $v\not\in N(w)$. As such, we have that
\begin{align*}
\Pr\paren{\text{$w$ not added to $K$}}&\leq \Pr\paren{X\geq 200\log{n}} = \Pr\paren{X\geq 2\cdot\expect{X}}\leq \frac{1}{n^{20}},
\end{align*}
where the last inequality is by an application of the Chernoff bound. Therefore, we could apply a union bound to the above events and conclude the proof.
\item For the second bullet, let us again define $X_v$ as the indicator random variable such that $v$ is a non-neighbor of $w$, and $X=\sum_{v\in T} X_v$ as the total number of non-neighbors in $T$. Since we have $\card{K\setminus N(w)}\geq 4\eps\cdot \card{K}$, we have that
\begin{align*}
\expect{X}&= \expect{\sum_{v\in T} X_v}\\
&= \sum_{v\in T}\expect{X_v} \tag{linearity of expectation}\\
&= \sum_{v \in T} \Pr\paren{v \not\in N(w)} \tag{by the definition of indicator random variable}\\
&\geq  100\cdot \frac{\log{n}}{\eps}\cdot 4\eps = 400 \cdot \log{n}.
\end{align*}
On the other hand, a necessary condition for $w$ to be added to $K$ is that $\card{N(w)\cap T}\geq (1-2\eps)\card{T}$, which means the non-neighbors of $w$ can be at most $2\eps\cdot 100\frac{\log{n}}{\eps} =200\log{n}$. As such, we can apply Chernoff bound again and show that the probability for $w$ to be added to $K$ is at most
\begin{align*}
\Pr\paren{\text{$w$ added to $K$}}&\leq \Pr\paren{X\leq 200\log{n}} = \Pr\paren{X\leq \frac{1}{2}\cdot \expect{X}} \leq \frac{1}{n^{20}},
\end{align*}
which is as desired by the second bullet of \Cref{lem:dense-merge-test}.
\end{itemize}
\end{proof}

\paragraph{An algorithm that creates sparse vertices.}
We now introduce a subroutine that tests a vertex (potentially in an almost-clique) is actually sparse.

\begin{Algorithm}
\label{alg:vertex-dense-test}
\textbf{An algorithm that splits $v$ from an almost-clique.}\\
\textbf{Input:} A graph $G=(V,E)$; a sparse-dense decomposition $\SDD_{G,\eps}$; a vertex $v\in K$ for some almost-clique $K$ in $\SDD_{G,\eps}$.
\smallskip
\begin{enumerate}
\item Sample a set $S(v)$ of $\min\{3000\cdot\frac{\log{n}}{\eps}, \deg(v)\}$ vertices from $N(v)$ uniformly at random.
\item Sample a set $D(v)$ of $\min\{3000\cdot\frac{\log{n}}{\eps}, \deg(v)\}$ vertices from $N(v)$ uniformly at random.
\item For each vertex $u \in D(v)$
\begin{enumerate}
\item Sample $\min\{3000\cdot\frac{\log{n}}{\eps}, \deg(u)\}$ vertices from $N(u)$ uniformly at random.
\item For each vertex $w\in S(u)\cup S(v)$, add $w$ to a set $\Delta(u,v)$ if $w\in N(u)\Delta N(v)$.
\item If $\card{\Delta(u,v)}\geq 41.5 \eps \cdot \card{S(u)\cup S(v)}$, then let $u$ be a \emph{sparse neighbor} of $v$.
\end{enumerate}
\item If $u$ has at least $40\eps \cdot \card{D(v)}$ vertices that are sparse neighbors, then return $v$ as a \emph{sparse vertex}.
\end{enumerate}
\end{Algorithm}

\Cref{alg:vertex-dense-test} serves as a ``counterpart'' of \Cref{alg:AC-dense-test}: it takes a graph partition (represented by a sparse-dense decomposition) and a vertex $v$ in on the of almost-cliques and decides whether to `peel off' the vertex from the almost clique. 
The following lemma characterizes the behavior of \Cref{alg:vertex-dense-test}. 
\begin{lemma}
\label{lem:sparse-split-test}
With high probability, \Cref{alg:vertex-dense-test} satisfies the following properties:
\begin{itemize}
\item If vertex $v$ is at least $42\eps$-sparse, then $v$ will be added to the set of sparse vertices.
\item If vertex $v$ is \emph{not} $38\eps$-sparse, then $v$ will \emph{not} be added to the set of sparse vertices.
\end{itemize}
\end{lemma}
\begin{proof}
Similar to the proof of \Cref{lem:dense-merge-test}, the proof of this lemma is another application of Chernoff bound. For the convenience of notation, for the rest of the proof, we let $\alpha = \card{N(v)\setminus N(u)}$ and $\beta = \card{N(u)\setminus N(v)}$. We also assume w.log. that $\deg(u)\geq \deg(v)$. We also assume that $\deg(v)\geq 3000\cdot \frac{\log{n}}{\eps}$ and $\deg(u)\geq 3000\cdot \frac{\log{n}}{\eps}$, since otherwise we can deterministically check the definition.

Note that by these definitions, we have that 
\begin{align*}
& \Pr(x\in N(v)\setminus N(u))=\frac{\alpha}{\deg(v)}\, \, \text{ for a random $x$ in $N(v)$}; \\
& \Pr(y\in N(u)\setminus N(v))=\frac{\beta}{\deg(u)} \, \, \text{ for a random $y$ in $N(u)$}.
\end{align*}
We first show that the estimator of $\card{\Delta(u, v)}$ is a good estimator for $\card{N(u)\sym N(v)}$ as follows.
\begin{enumerate}[label=(\roman*).]
\item If $\card{N(u)\sym N(v)}\geq 42\eps \cdot \max\{\deg(u), \deg(v)\}$, then with high probability, $\card{\Delta(u, v)}\geq 40\eps\cdot \max\{\card{S(u)}, \card{S(s)}\}$.  
To see this, we let $X_i$ be a random variable which is $1$ when the $i^{th}$ sample in $S(v)$ is in $N(v)\setminus N(u)$.
Let $Y_i$ be a random variable which is $1$ when the $i^{th}$ sample in $S(u)$ is in $N(u)\setminus N(v)$.
Let $Z_i = X_i+Y_i$.
The expected value of $Z_i$ is the following:
\begin{align*}
\expect{Z_i} &= \expect{X_i} + \expect{Y_i} \\
& = \Pr(X_i=1) + \Pr(Y_i=1)\\
& = \frac{\alpha}{\deg(v)} + \frac{\beta}{\deg(u)}\\
& \geq \frac{\alpha+\beta}{\deg(u)}, \tag{we assume w.log. that $\deg(u)\geq \deg(v)$}
\end{align*}
 Since we have $\alpha+\beta \geq 42\eps \max\{\deg(u),\deg(v)\}$, we have that $\expect{Z_i} \geq 42\eps$. Therefore, if we define $Z=\sum_i Z_i$, we have that 
\begin{align*}
\expect{Z} & = \sum_{i=1}^{(3000 \log n)/\eps} \expect{Z_i} \\
& \geq \frac{3000\log{n}}{\eps}\cdot 42\eps = 126000\log{n}.
\end{align*}
Since $Z$ is a summation of independent indicator random variables, we have that 
\begin{align*}
\Pr\paren{Z\leq 41.5\cdot 3000 \log n}& \leq \Pr\paren{Z\leq (1-1/84)\cdot \expect{Z}}\\
& \leq \exp\paren{-\frac{(1/84)^2\cdot 126000\log{n}}{3}} \leq \frac{1}{n^{5}},
\end{align*}
and the second-last inequality is due to the multiplicative Chernoff bound.
\item If $\card{N(u)\sym N(v)}< 38\eps \cdot \max\{\deg(u), \deg(v)\}$, then with high probability, $\card{\Delta(u, v)}< 40\eps\cdot \max\{\card{S(u)}, \card{S(s)}\}$. We use the same random variables $X_i,Y_i,Z_i$ as the previous case and have that
\begin{align*}
\expect{Z_i} & = \frac{\alpha}{\deg(v)} + \frac{\beta}{\deg(u)}\\
& \leq \frac{\alpha+\beta}{\deg(v)}. \tag{we assume w.log. that $\deg(u)\geq \deg(v)$}
\end{align*}

Note that in this case, we should have $\deg(u)< (1+40\eps)\cdot \deg(v)$, otherwise we get a contradiction to $\card{N(u)\sym N(v)}< 38\eps \cdot \max\{\deg(u), \deg(v)\} = 38\eps \cdot \deg(u)$.

Assume towards a contradiction $\deg(u)= (1+\delta)\cdot \deg(v)$ where $\delta\geq 40\eps$. This implies that
\begin{align*}
\card{N(u)\sym N(v)}\geq \deg(u)- \deg(v) = \delta \deg(v) = \frac{\delta}{1+\delta} \deg(u) \geq \frac{40\eps}{1+40 \eps} \deg(u),
\end{align*}
since $\frac{\delta}{1+\delta}$ is an increasing function when $\delta\geq 0$.
Using $\eps<1/2000$ we get $\card{N(u)\sym N(v)} \geq 38\eps \cdot \deg(u)$ giving us a contradiction.
Therefore, we could further upper-bound the probability as
\begin{align*}
\expect{Z_i} \leq \frac{\alpha+\beta}{\deg(v)} < \frac{38 \eps \cdot \deg(u)}{\deg(v)} \leq 38\eps \cdot (1+40\eps) \leq 41\eps,
\end{align*}
where the last inequality holds for $\eps<1/2000$.
Similar to the previous case, we define $Z=\sum_{i} Z_i$ and have that 
\begin{align*}
\expect{Z} & = \sum_{i=1}^{(3000 \log n)/\eps} \expect{Z_i} \\
& \leq \frac{3000\log{n}}{\eps}\cdot 41\eps = 123000\log{n}.
\end{align*}
Again, since $Z$ is a summation of independent indicator random variables, we have that 
\begin{align*}
\Pr\paren{Z\geq 41.5\cdot 3000 \log n}& \leq \Pr\paren{Z\leq (1+1/82)\cdot \expect{Z}}\\
& \leq \exp\paren{-\frac{(1/82)^2\cdot 123000\log{n}}{3}} \leq \frac{1}{n^{5}},
\end{align*}
and the second-last inequality is due to the multiplicative Chernoff bound.
\end{enumerate}

We now prove the two bullet points in the statement of \Cref{lem:sparse-split-test} in order.
\begin{itemize}
\item For the first bullet, we note that if the vertex $v$ is at least $42\eps$-sparse then at least $42\eps$ fraction of $v$'s neighbors $v_i$ satisfy $\card{N(v_i)\sym N(v)}\geq 42\eps \cdot \max\{\deg(v_i), \deg(v)\}$ which implies that $v_i$ is a ``sparse neighbor'' of $v$ (we already showed this above).
Consider a vertex $u \in D(v)$. The probability that $u$ is a sparse neighbor of $v$ is at least $42\eps$. Thus out of $\card{D(v)}$ samples we expect at least $42 \eps \card{D(v)}$ samples to be sparse neighbors. Let $X$ be the random variable denoting the number of sparse neighbors. 
Since $X$ is a summation of independent random variables, we can apply the Chernoff bound to obtain that
\begin{align*}
\Pr\paren{X\leq 40\eps \cdot \card{D(v)}} &\leq \exp\paren{-\frac{(1/21)^2 \cdot 42 \cdot 3000 \cdot \log{n}}{3}}\leq \frac{1}{n^{10}}.
\end{align*}
As such, we reach the desired conclusion of the first bullet.

\item For the second bullet, we note that if the vertex $v$ is not $38\eps$-sparse then at most $38\eps$ fraction of $v$'s neighbors $v_i$ satisfy $\card{N(v_i)\sym N(v)}\geq 38\eps \cdot \max\{\deg(v_i), \deg(v)\}$ which implies that $v_i$ is a ``sparse neighbor'' of $v$ (we already showed this above).
Consider a vertex $u \in D(v)$. The probability that $u$ is a sparse neighbor of $v$ is at most $38\eps$. Thus out of $\card{D(v)}$ samples we expect at most $38 \eps \card{D(v)}$ samples to be sparse neighbors. Let $X$ be the random variable denoting the number of sparse neighbors. 
Since $X$ is a summation of independent random variables, we can apply the Chernoff bound to obtain that
\begin{align*}
\Pr\paren{X\geq 40\eps \cdot \card{D(v)}} &\leq \exp\paren{-\frac{(1/20)^2 \cdot 38 \cdot 3000 \cdot \log{n}}{3}}\leq \frac{1}{n^{10}}.
\end{align*}
Thus, we reach the desired conclusion of the second bullet.
\end{itemize}
\myqed{\Cref{lem:sparse-split-test}}
\end{proof}

%% file: experiments.tex
\section{Experimental Evaluation}
\label{sec:experiments}
To demonstrate the practical validity of our algorithm, we evaluate our approach on synthetic as well as real world datasets. 

\subsection*{Setup and Datasets}

We evaluate and compare the correlation clustering cost for dynamic edge updates on synthetic and real world graph of our algorithm against pivot. Note that for edge insertions and deletions against an adaptive adversary, there exists no dynamic algorithm in the literature, hence our comparison is against running pivot after each update using fresh randomness (this by default is adversarially robust but needs $O(n)$ time after each update). For our dynamic SDD implementation, we initialize by marking every vertex as sparse. The algorithm requires setting the $\eps$ parameters, which decides when to trigger an update and also controls the sample size used by subroutines. Empirically, we search for $\eps$ in range $[0.3,0.6]$ in increments of $0.025$ that gives the lowest cost for a small number of updates and use that for our experiments. The $\eps$ parameter used for different settings of our experiments are listed in \Cref{table:exp_eps_values}. 
In practice, we could use the meta-information we have of the instances (e.g., how dense the graph is), and pick parameters accordingly.

\paragraph{Compute resources:} All experiments were run on Apple Macbook Pro with M1 processor and 16GB RAM.

For synthetic datasets, we use the widely studied stochastic block model (SBM) for generating the graph and clusters (communities) within. For a graph on $n$ vertices, SBMs are paramerized by number of clusters $(k)$, the probability of an edge within a cluster $(p)$ and probability of an edge across clusters $(q)$. For our experiments, we generate SBM with $k = \{4, 10\}$ with $p=0.95$ and $q=0.05$ for $n = \{250, 2000\}$ resapectively. For real world datasets, we use two graphs from SNAP \cite{snapnets} email network (email-Eu-core with $n=1005, m=25571$) and collaboration network (ca-HepTh with $n=9877, m=25998$). We use the adjacency list data structure for storing and updating the graph.

\begin{table}[]
    \centering
    \begin{small}
\begin{tabular}{ccc}
     \toprule
     Graph & $\eps$ for random updates & $\eps$ for targeted updates \\
     \midrule
     SBM $(n=250, k=4)$ & 0.45 & 0.4  \\
     SBM $(n=250, k=10)$ & 0.45 & 0.4 \\
     SBM $(n=2000, k=4)$ & 0.45 & 0.45  \\
     SBM $(n=2000, k=10)$ & 0.425 & 0.45  \\
     email-Eu-core & 0.45 & 0.45  \\
     ca-HepTh & 0.45 & 0.45  \\
     \bottomrule
\end{tabular}
    \caption{Value of $\eps$ parameter used for }
    \label{table:exp_eps_values}
    \end{small}
\end{table}

\paragraph{Simulating the edge updates:}

For each dataset, we simulate two types of edge updates, random and targeted. The targeted edge updates are meant to simulate an adversarial scenario. In random edge updates, we start by fixing a permutation on the edges of the graph (edges that are present in the simulated graph or the actual real-world graph). Starting with an empty graph, we simulate updates by going over the edge permutation. For each new update, we first toss a coin with a fixed probability and insert the next edge in the permutation if heads and delete a random edge (henceforth referred to as deletion probability) from the current state of the graph with tails. In our experiments, we fix the probability of deletion to be $0.2$ for all experiments. For random updates, we simulate for $\max(1.5*m, 500k)$ updates, where $m$ is the number of edges in the original graph.

For targeted updates, we first simulate a random stream for the entire graph with $0$ deletion probability. The goal of targeted updates is to find the two biggest clusters and `destroy' their structure. Towards this end, we use the current pivot clustering to identify the two largest clusters in the graph. After identifying, the next batch of edge updates is selected as (i) deleting the edges within the two respective clusters, and (ii) inserting new edges across the vertices of the two clusters. We use the same deletion probability of $0.2$. Once the batch of updates is exhausted, i.e., we have finished deleting all the within-cluster edges and adding across-cluster edges, the process is repeated with the two new biggest clusters in the current graph. We simulate this for $\max(2.5*m, 500k)$ updates, where the first $m$ updates are the random stream to fill in the graph.

We compare the cost of dynamic SDD and pivot with singleton clustering across updates. In Singleton, each vertex forms its own cluster. The cost reported is the ratio of costs, i.e. cost of SDD of pivot clustering divided by the cost of singleton clustering. We report cost after every 100 updates for random updates and for targeted updates, every 100 updates after the first $(1/2)*m$ updates are done (this is because the trend for these updates will only be similar to random updates).

\subsection*{Summary of experiments}

For random updates on SBM, the experiments start with an empty graph where edges are added or deleted, leading to a sparse graph in the initial stages with small clusters. Pivot greedily creates clusters; as a result, it has worse cost, since keeping vertices as singleton clusters incurs less cost. Since there are no edges at the start, as per the SDD definition, all vertices are sparse and treated as singleton clusters. As more edges are inserted, with our setting of $p=0.95$, the clusters are dense and pivot starts to perform better, but even so, our algorithm consistently maintains a clustering with lower cost, as can be seen in \Cref{fig:sbm_250_4_rand,fig:sbm_250_10_rand,fig:sbm_2000_4_rand,fig:sbm_2000_10_rand}.

For targeted updates on SBM, the updates are simulated in such a way that destroy the clusters and instead make the subgraph restricted to those clusters complete bipartite. 
Here we observe that for small number of clusters $(k=4)$ our algorithm outputs clustering with better cost \Cref{fig:sbm_250_4_adv,fig:sbm_2000_4_adv}, and for more clusters $(k=10)$ the clustering given by our algorithm has more stable cost, whereas for pivot the cost fluctuates a lot, even being worse than singleton clustering \Cref{fig:sbm_250_10_adv,fig:sbm_2000_10_adv}.

Experiments on real-world datasets have a similar trend. Since real-world graphs are often sparse, pivot ends up incurring higher cost for both random and targeted updates, and the cost tends to fluctuate. On the other hand, the cost of our algorithm tends to be very stable and smaller than pivot (\Cref{fig:email_rand,fig:email_adv,fig:hepth_rand,fig:hepth_adv}).
For random updates, it appears that singleton clustering gives relatively low costs (i.e., it is hard for both algorithms to significantly outperform Singleton).
We believe this is because the real-world graphs we picked are very sparse, and testing the algorithms on dense large-scale real-world graphs can be an interesting direction for future explorations.

\FloatBarrier

\begin{figure}[!htb]
  \centering
  \begin{subfigure}[t]{.24\linewidth}
    \centering\includegraphics[width=\linewidth]{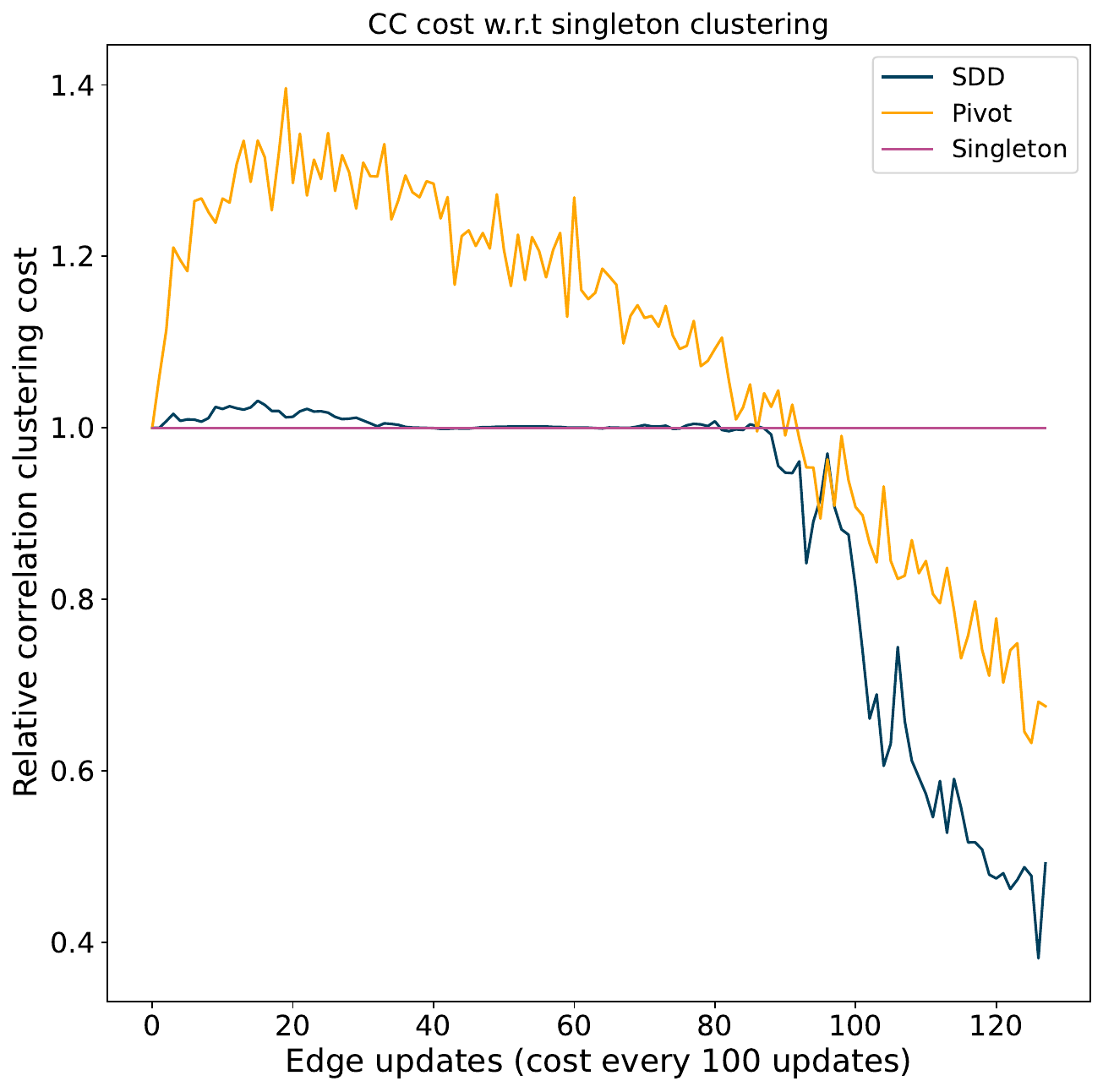}
    \caption{Random updates}
    \label{fig:sbm_250_4_rand}
  \end{subfigure}
  \begin{subfigure}[t]{.24\linewidth}
    \centering\includegraphics[width=\linewidth]{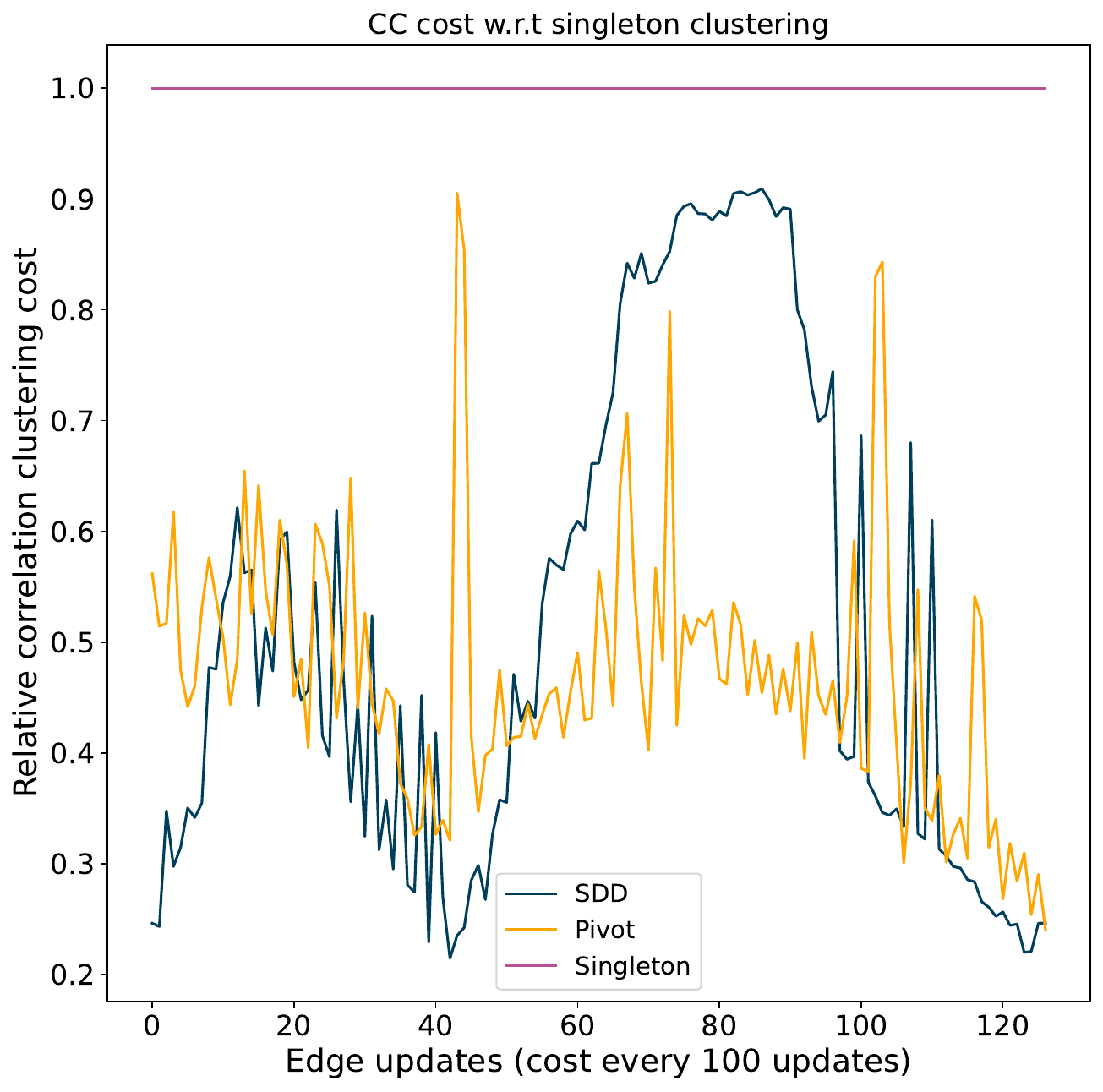}
    \caption{Targeted updates}
    \label{fig:sbm_250_4_adv}
  \end{subfigure}
  \begin{subfigure}[t]{.24\linewidth}
    \centering\includegraphics[width=\linewidth]{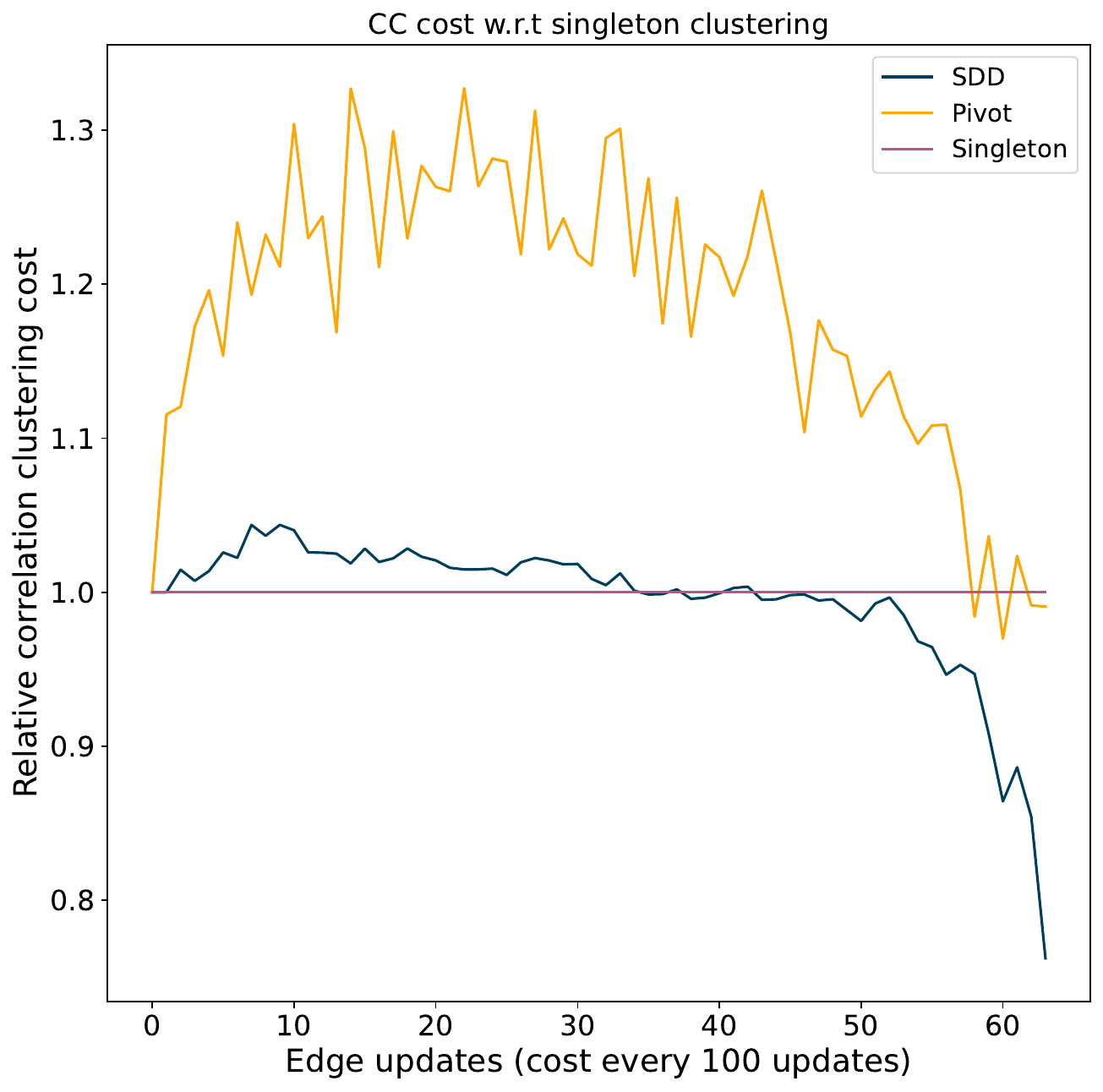}
    \caption{Random updates}
    \label{fig:sbm_250_10_rand}
  \end{subfigure}
  \begin{subfigure}[t]{.24\linewidth}
    \centering\includegraphics[width=\linewidth]{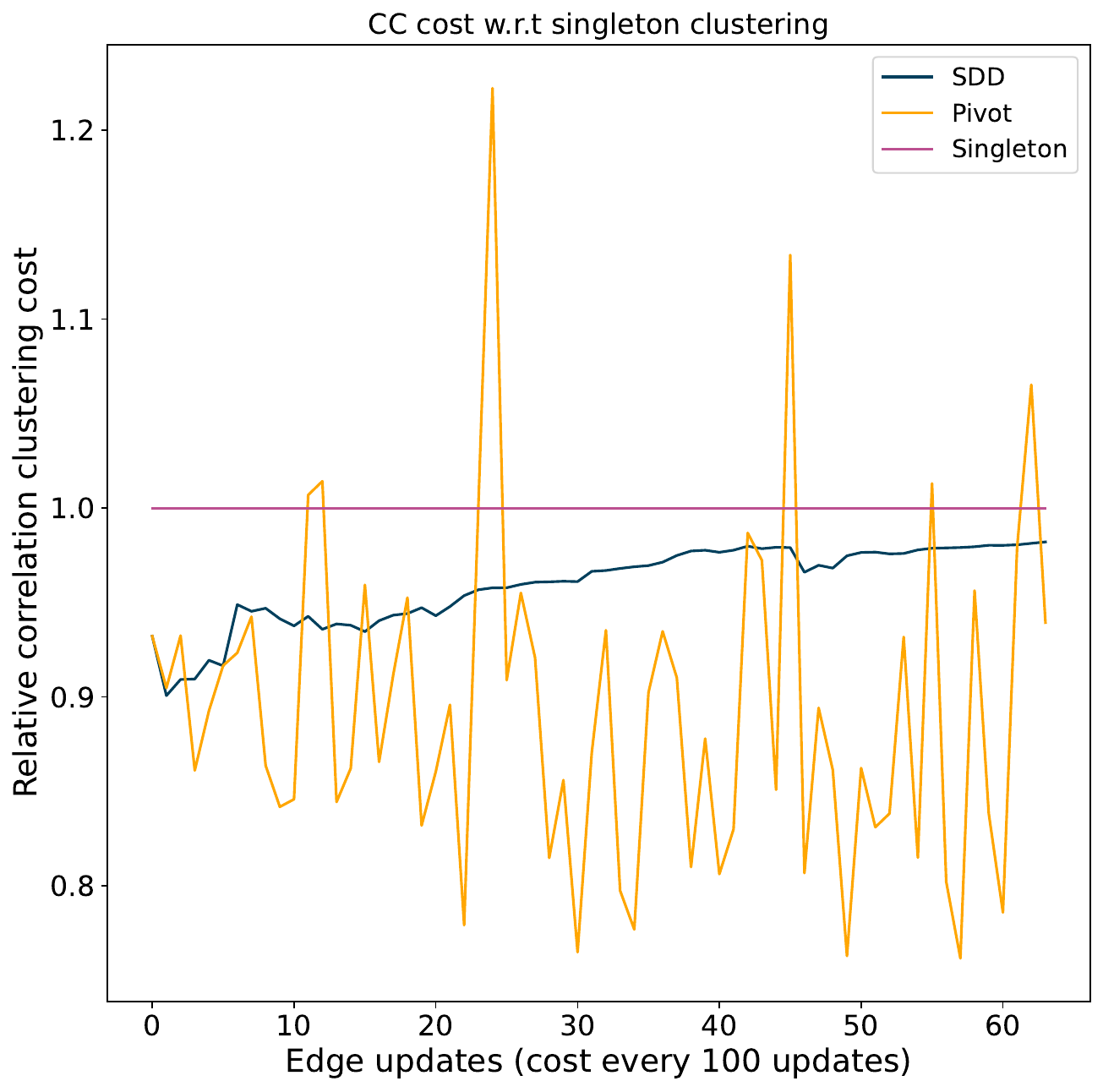}
    \caption{Targeted updates}
    \label{fig:sbm_250_10_adv}
  \end{subfigure}
  \caption{Plot of CC cost ratio w.r.t singletons for SBM: $n=250$ and $k=4$ (a,b) and $k=10$ (c,d).}
\end{figure}

\begin{figure}[!htb]
  \centering
  \begin{subfigure}[t]{.24\linewidth}
    \centering\includegraphics[width=\linewidth]{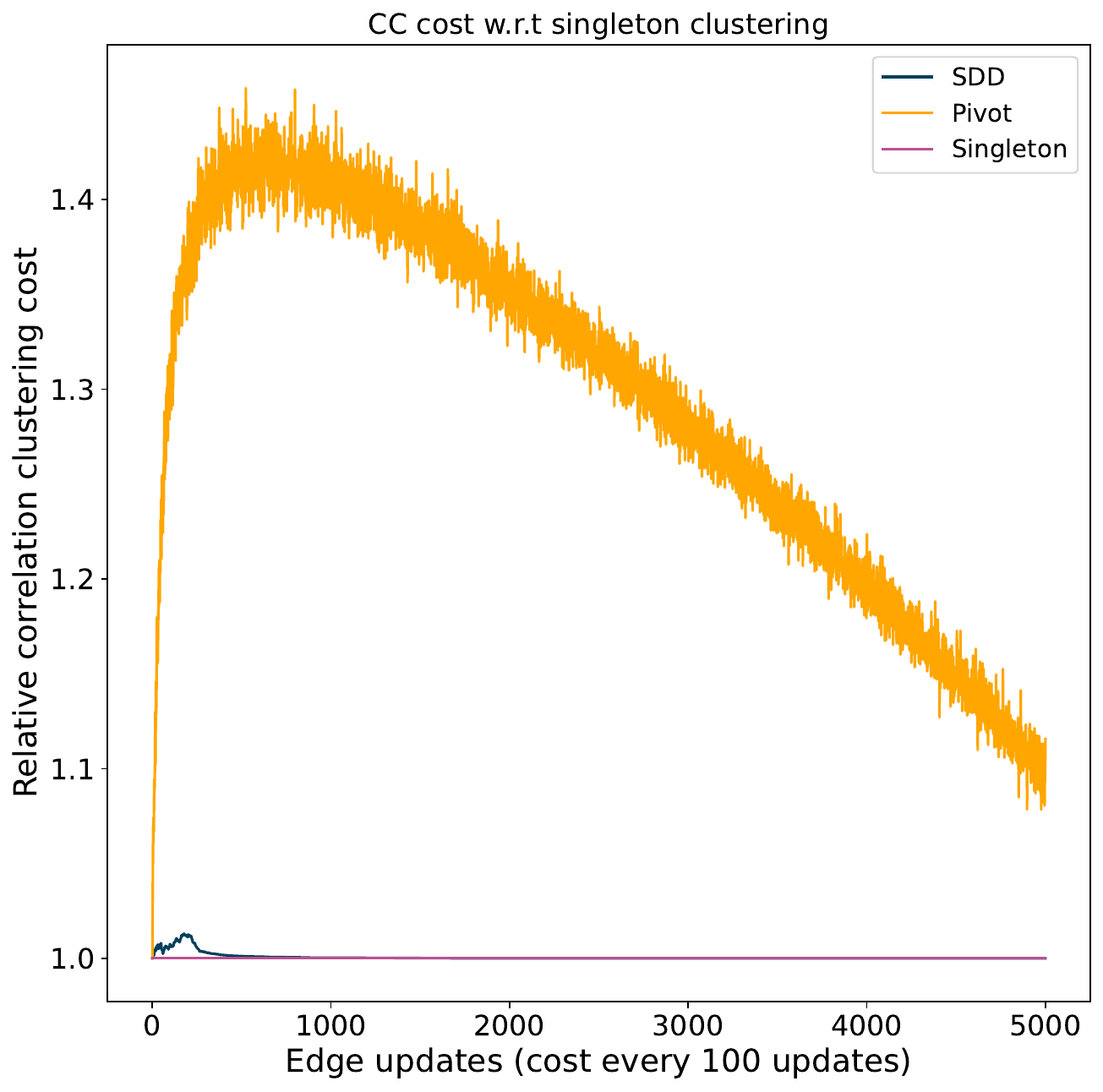}
    \caption{Random updates}
    \label{fig:sbm_2000_4_rand}
  \end{subfigure}
  \begin{subfigure}[t]{.24\linewidth}
    \centering\includegraphics[width=\linewidth]{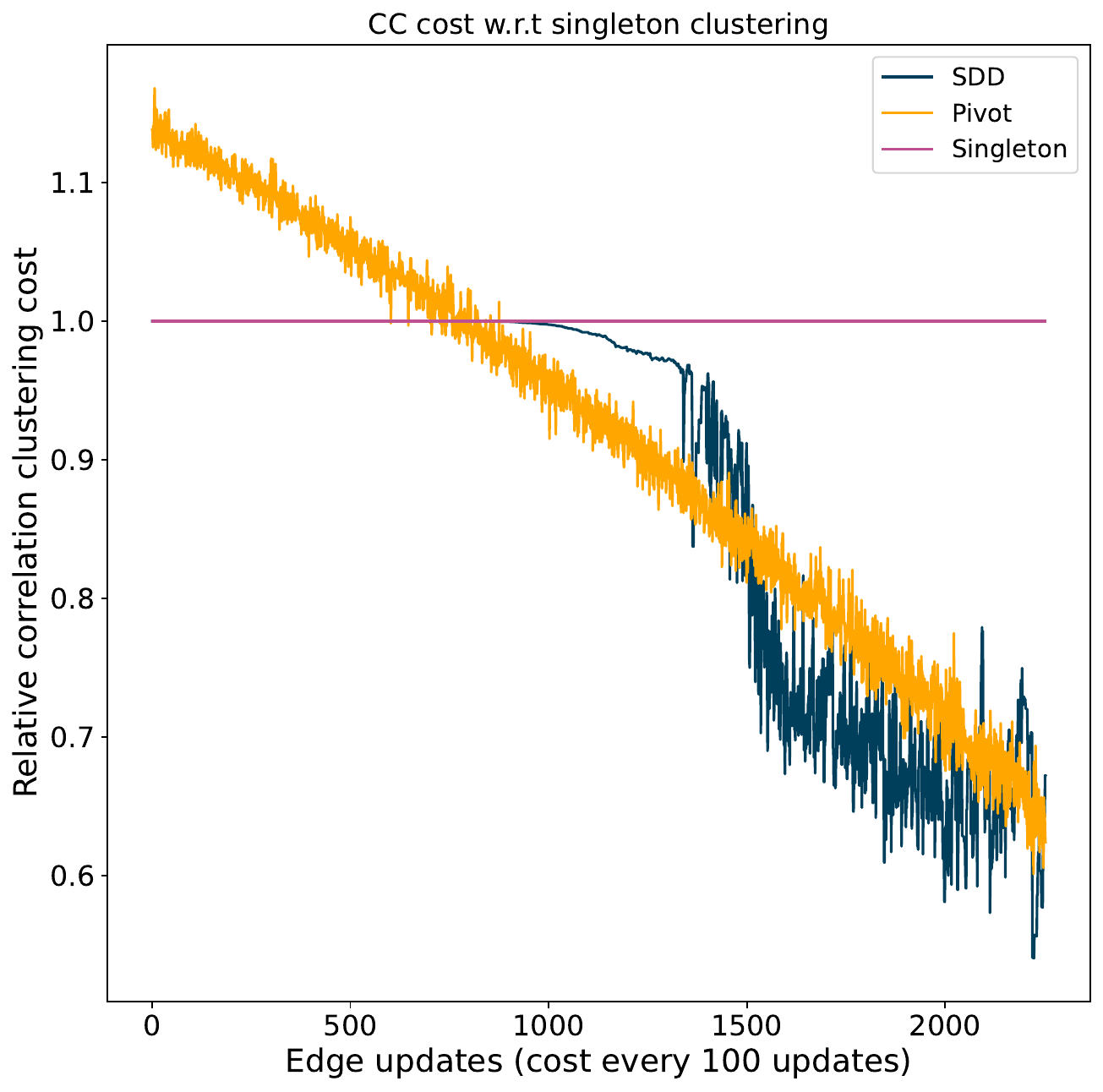}
    \caption{Targeted updates}
    \label{fig:sbm_2000_4_adv}
  \end{subfigure}
  \begin{subfigure}[t]{.24\linewidth}
    \centering\includegraphics[width=\linewidth]{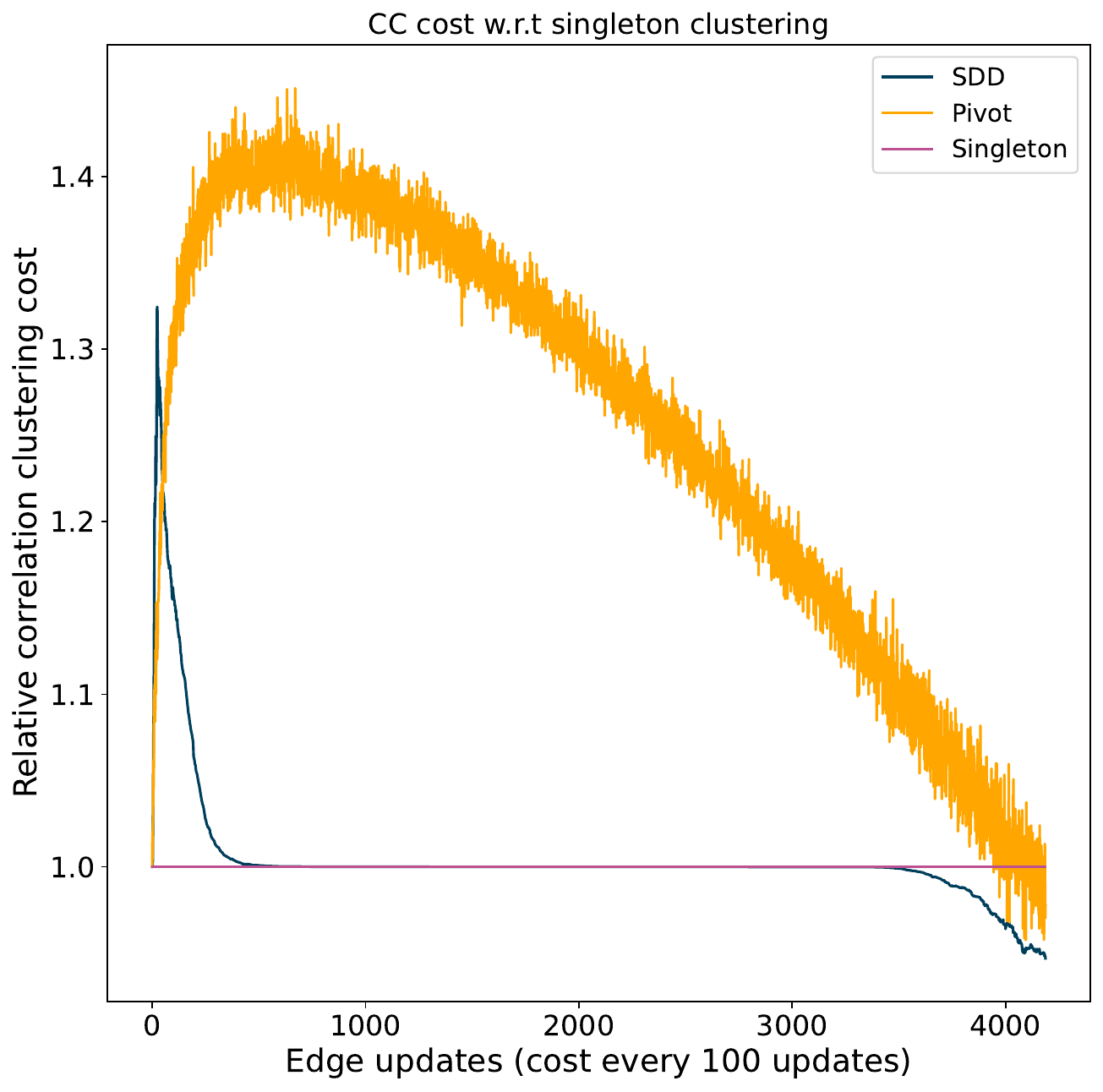}
    \caption{Random updates}
    \label{fig:sbm_2000_10_rand}
  \end{subfigure}
  \begin{subfigure}[t]{.24\linewidth}
    \centering\includegraphics[width=\linewidth]{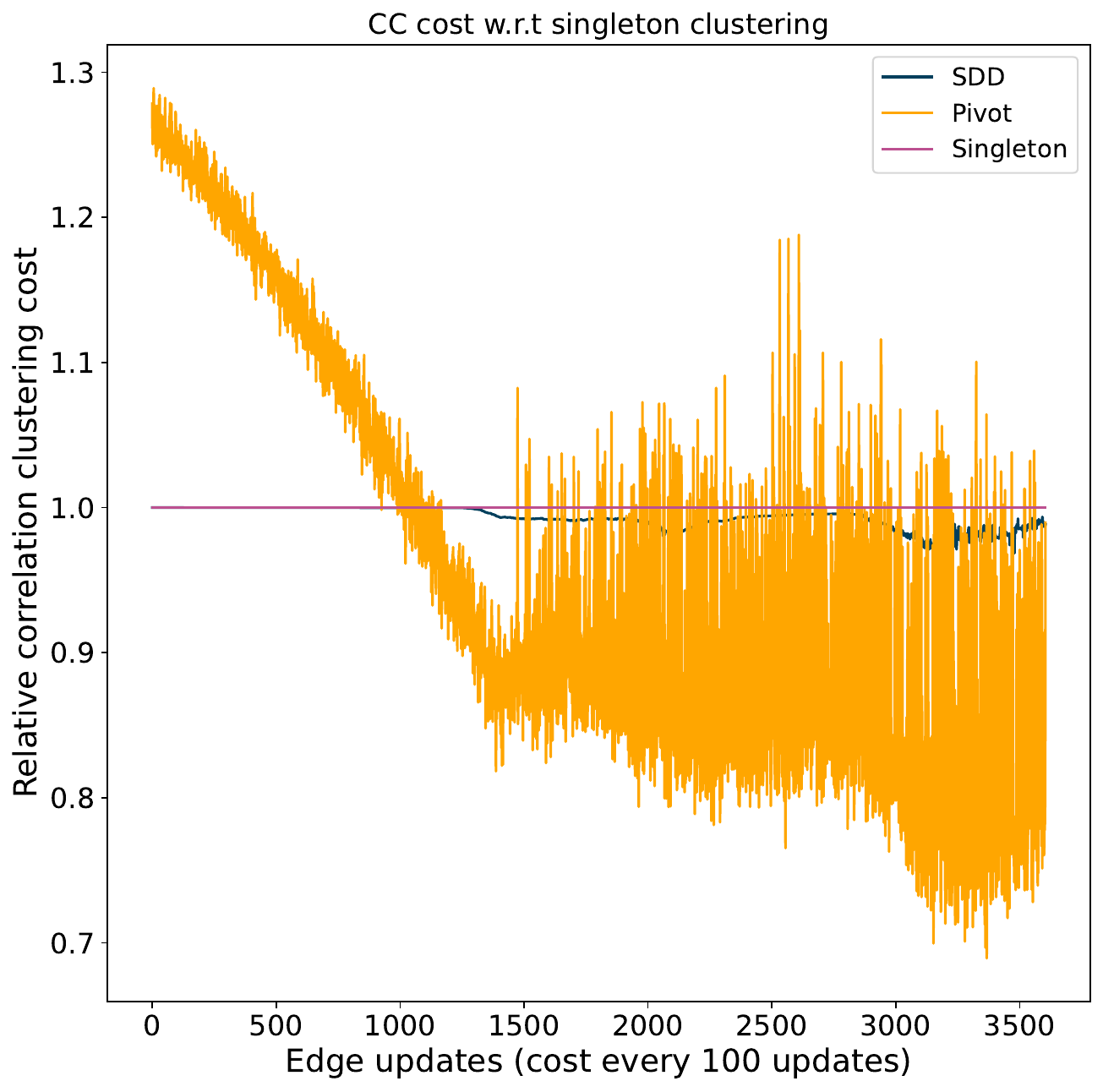}
    \caption{Targeted updates}
    \label{fig:sbm_2000_10_adv}
  \end{subfigure}
  \caption{Plot of CC cost ratio w.r.t singletons for SBM: $n=2000$ and $k=4$ (a,b) and $k=10$ (c,d).}
\end{figure}

\begin{figure}[!htb]
  \centering
  \begin{subfigure}[t]{.24\linewidth}
    \centering\includegraphics[width=\linewidth]{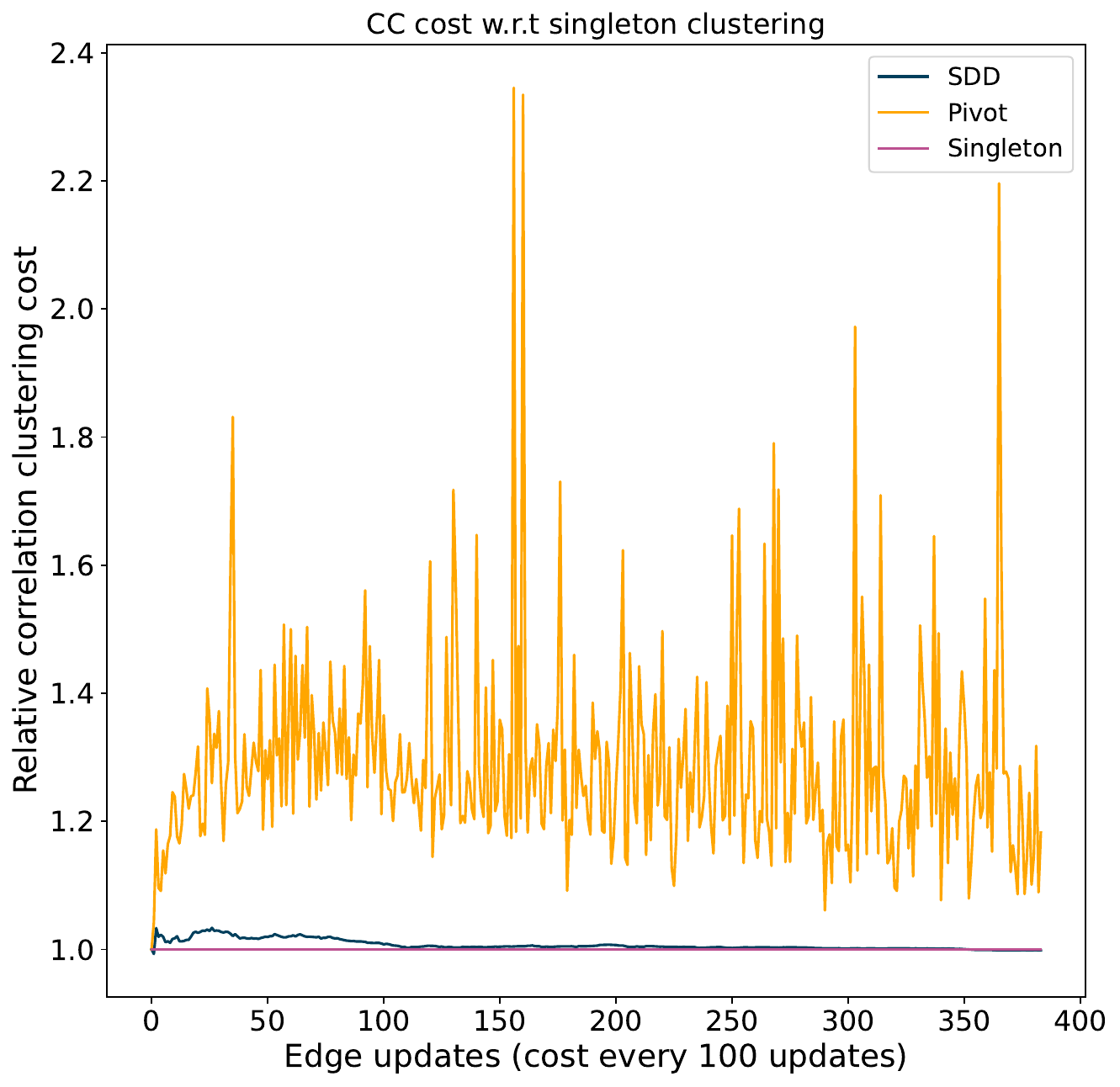}
    \caption{Random updates}
    \label{fig:email_rand}
  \end{subfigure}
  \begin{subfigure}[t]{.24\linewidth}
    \centering\includegraphics[width=\linewidth]{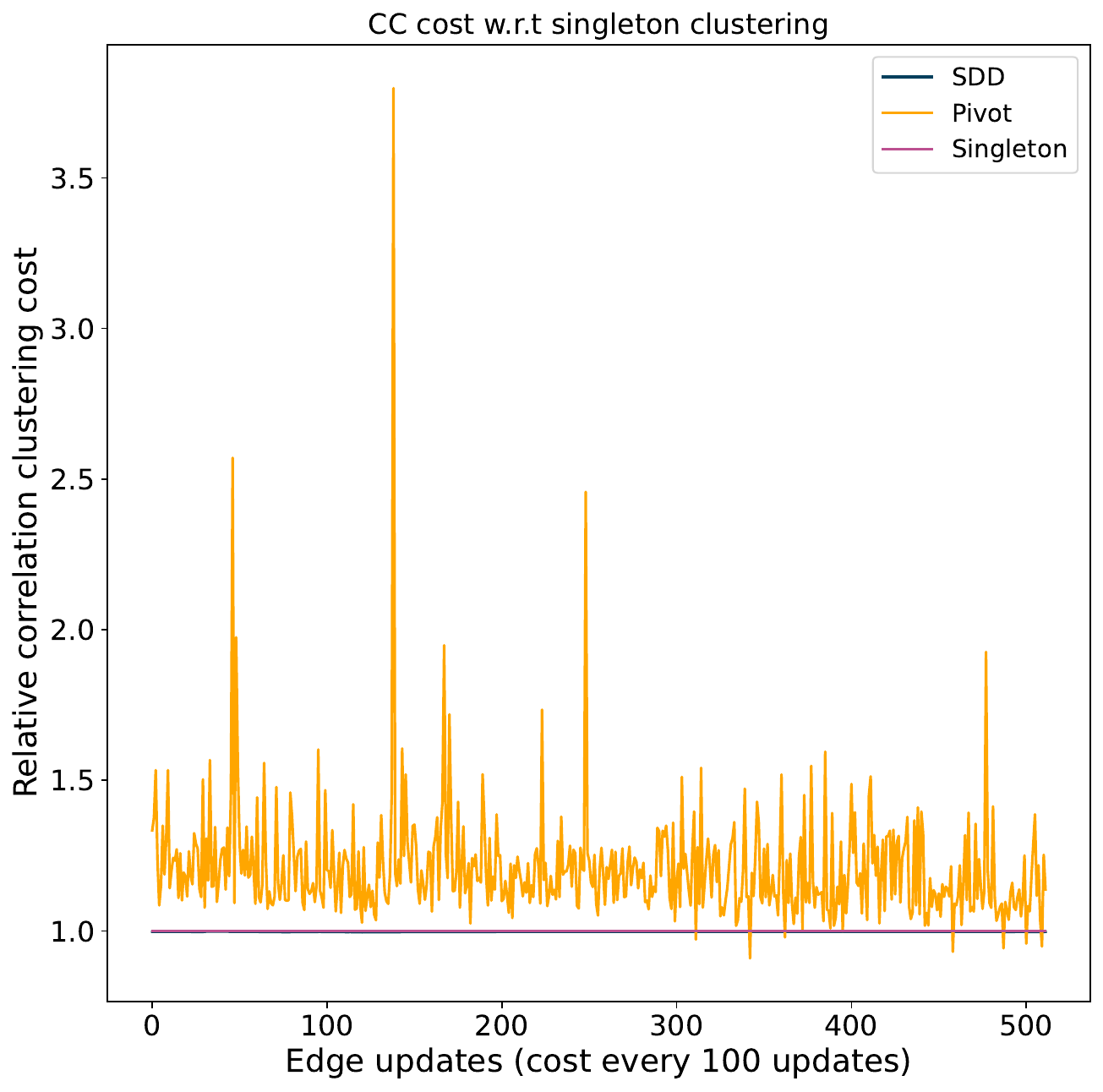}
    \caption{Targeted updates}
    \label{fig:email_adv}
  \end{subfigure}
  \begin{subfigure}[t]{.24\linewidth}
    \centering\includegraphics[width=\linewidth]{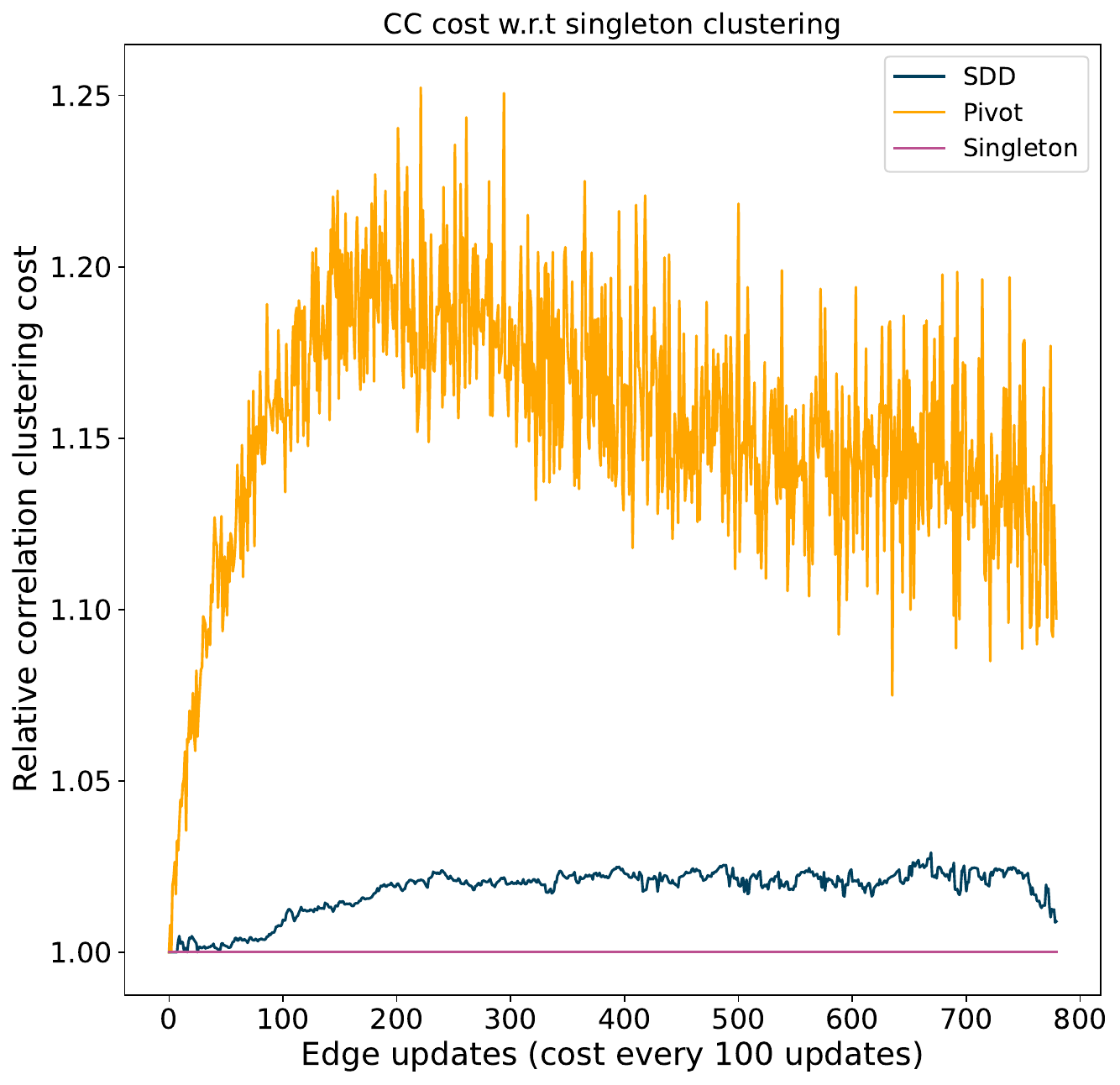}
    \caption{Random updates}
    \label{fig:hepth_rand}
  \end{subfigure}
  \begin{subfigure}[t]{.24\linewidth}
    \centering\includegraphics[width=\linewidth]{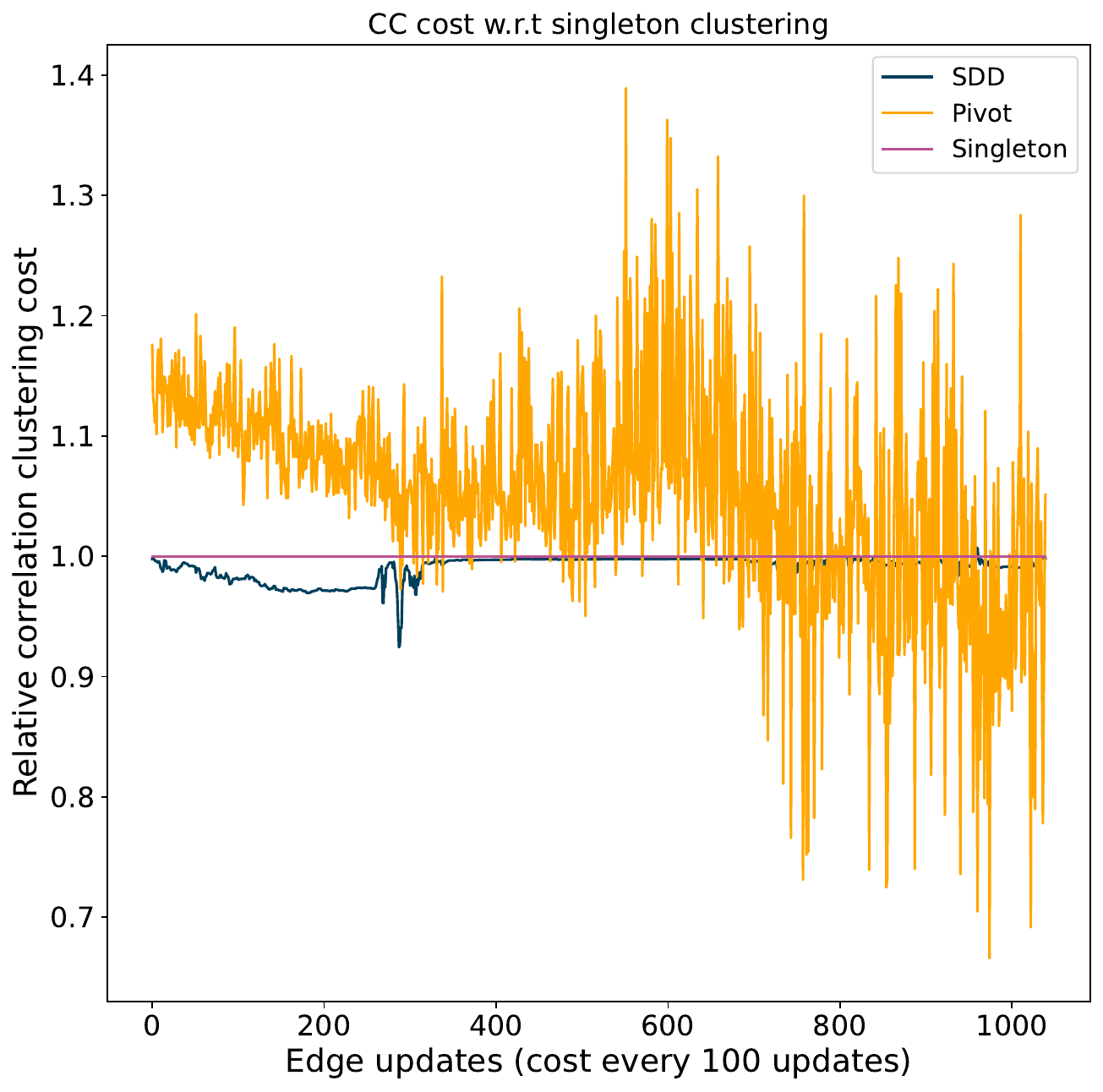}
    \caption{Targeted updates}
    \label{fig:hepth_adv}
  \end{subfigure}
  \caption{Plot of CC cost ratio w.r.t singletons for eemail-Enron (a,b) and ca-AstroPh (c,d).}
\end{figure}